\documentclass[a4paper,11pt]{article}

\usepackage[
margin=1in,
includefoot,
footskip=30pt,
]{geometry}
\usepackage{xspace}
\usepackage{tikz}
\usepackage{svg}
\usepackage{mathrsfs,amsmath,amsthm,amssymb,mdwlist,bbm}

\usepackage{thmtools} 
\usepackage{thm-restate}
\usepackage{color}
\definecolor{antiquebrass}{rgb}{0.8, 0.58, 0.46}
\usepackage{algpseudocode,algorithm}

\algrenewcommand\algorithmicrequire{\textbf{Input:}}
\algrenewcommand\algorithmicensure{\textbf{Output:}}
\algnewcommand{\Initialize}[1]{%
  \State \textbf{Initialize:}
  \Statex \hspace*{\algorithmicindent}\parbox[t]{.8\linewidth}{\raggedright #1}
}
\algdef{SE}[FUNCTION]{Function}{EndFunction}%
   [2]{\algorithmicfunction\ \textproc{#1}\ifthenelse{\equal{#2}{}}{}{(#2)}}%
   {\algorithmicend\ \algorithmicfunction}%

\usepackage{subcaption}

\usepackage{makecell}
\usepackage{multirow}

\allowdisplaybreaks

\renewcommand{\leq}{\leqslant}
\renewcommand{\geq}{\geqslant}
\renewcommand{\ge}{\geqslant}
\renewcommand{\le}{\leqslant}

\usepackage{enumerate}
\usepackage{authblk}
\usepackage[hidelinks]{hyperref}
\usepackage[shortlabels]{enumitem}

\DeclareMathOperator*{\argmax}{arg\,max}
\newcommand{\opt}{\mathsf{opt}}
\newcommand{\OPT}{\mathsf{OPT}}
\newcommand{\alg}{\mathsf{ALG}}
\newcommand{\range}{\mathsf{range}}
\newcommand{\greedy}{\mathsf{Greedy}}
\newcommand{\In}{\mathsf{In}}
\newcommand{\Out}{\mathsf{Out}}

\newcommand{\pr}{\mathsf{Pr}}

\newtheorem{lemma}{Lemma}
\newtheorem{definition}{Definition}

 \setlist{nolistsep,leftmargin=*}

\usepackage{nicefrac}

\newcommand{\eps}{\ensuremath{\varepsilon}\xspace}
\renewcommand{\epsilon}{\eps}
\let\mydelta\delta
\renewcommand{\delta}{\ensuremath{\mydelta}\xspace}
\let\myalpha\alpha
\renewcommand{\alpha}{\ensuremath{\myalpha}\xspace}
\let\mystar\star
\renewcommand{\star}{\ensuremath{\mystar}}

\newcommand{\DD}{\ensuremath{\mathcal D}\xspace}

\newcommand{\FF}{\ensuremath{\mathcal F}\xspace}

\newcommand{\HH}{\ensuremath{\mathcal H}\xspace}
\newcommand{\II}{\ensuremath{\mathcal I}\xspace}

\renewcommand{\SS}{\ensuremath{\mathcal S}\xspace}

\newtheorem{theorem}{\bf Theorem}
\newtheorem{proposition}{\bf Proposition}
\newtheorem{observation}{\bf Observation}

\newtheorem{claim}{\bf Claim}

\usepackage{cleveref}

\crefname{theorem}{Theorem}{Theorems}
\crefname{observation}{Observation}{Observations}
\crefname{lemma}{Lemma}{Lemmas}
\crefname{corollary}{Corollary}{Corollaries}
\crefname{proposition}{Proposition}{Propositions}
\crefname{example}{Example}{Examples}
\crefname{claim}{Claim}{Claims}
\crefname{table}{Table}{Tables}
\crefname{equation}{Inequality}{Inequalities}
\crefname{reductionrule}{Reduction rule}{Reduction rules}
\crefname{section}{Section}{Sections}
\crefname{appendix}{Appendix}{Appendixes}

\title{Random-Order Online Independent Set of Intervals and Hyperrectangles}
\author{Mohit Garg\thanks{Supported by a fellowship from the Walmart Center for Tech Excellence at IISc (CSR Grant WMGT-23-0001).}}
\author{Debajyoti Kar\thanks{Supported by the Google PhD Fellowship.}}
\author{Arindam Khan\thanks{
Supported in part by Google India Research Award, SERB Core Research Grant (CRG/2022/001176) on “Optimization under Intractability and Uncertainty”, and the Walmart Center for Tech Excellence at IISc (CSR Grant WMGT-23-0001).}}
\affil{Indian Institute of Science, Bengaluru}
\affil{\texttt {\{mohitgarg, debajyotikar, arindamkhan\}@iisc.ac.in}}
\date{}

\begin{document}

\maketitle

%%%%%%%%%%ABSTRACT%%%%%%%%%%%%
\begin{abstract}
    In the {\sc Maximum Independent Set of Hyperrectangles} problem, we are given a set of $n$ (possibly overlapping) $d$-dimensional axis-aligned hyperrectangles, and the goal is to find a subset of non-overlapping hyperrectangles of maximum cardinality. For $d=1$, this corresponds to the classical {\sc Interval Scheduling} problem, where a simple greedy algorithm returns an optimal solution. In the offline setting, for $d$-dimensional hyperrectangles, 
polynomial time  $(\log n)^{O(d)}$-approximation algorithms are known \cite{chalermsook2009maximum}. However, the problem becomes notably challenging in the online setting, where the input objects (hyperrectangles) appear one by one in an adversarial order, and on the arrival of an object, the algorithm needs to make an immediate and irrevocable decision whether or not to select the object while maintaining the feasibility. Even for interval scheduling, an $\Omega(n)$ lower bound is known on the competitive ratio. 

To circumvent these negative results, in this work, we study the online maximum independent set of axis-aligned hyperrectangles in the random-order arrival model, where the adversary specifies the set of input objects which then arrive in a uniformly random order. Starting from the prototypical secretary problem, the random-order model has received significant attention to study algorithms beyond the worst-case competitive analysis (see the survey by Gupta and Singla \cite{Gupta020}). Surprisingly, we show that the problem in the random-order model almost matches the best-known offline approximation guarantees, up to polylogarithmic factors. In particular, we give a simple $(\log n)^{O(d)}$-competitive algorithm for $d$-dimensional hyperrectangles in this model, which runs in $\tilde{O_d}(n)$ time. Our approach also yields $(\log n)^{O(d)}$-competitive algorithms in the random-order model for more general objects such as $d$-dimensional fat objects and ellipsoids. Furthermore, all our competitiveness guarantees hold with high probability, and not just in expectation.
\end{abstract}
%%%%%%%End of ABSTRACT%%%%%%%%

%%%%%%%%%%Introduction%%%%%%%%
\section{Introduction}
Geometric intersection graphs are extensively studied in discrete geometry, graph theory, and theoretical computer science. In the intersection graph of a collection of geometric objects, the objects form the vertices, and two vertices are connected by an edge if and only if the corresponding objects intersect.
The maximum independent set problem for geometric intersection graphs is a fundamental problem in computational geometry. In an equivalent geometric formulation of the problem,
given a set $\HH$ of $n$ possibly overlapping geometric objects in $d$-dimensional Euclidean space, our goal is to select a subset ${\HH}^* \subseteq \HH$ of maximum cardinality such that the objects in ${\HH}^*$ are pairwise disjoint.
Apart from its practical applications in data mining \cite{fukuda1996data}, map labeling \cite{doerschler1992rule}, routing \cite{lewin2002routing}, etc., studying the problem has led to the development of several techniques in geometric approximation algorithms \cite{har2008geometric}.

For interval graphs (i.e., when the objects are intervals), the problem is also known as interval scheduling, where a simple greedy algorithm is optimal. In contrast, the problem is W[1]-hard (with respect to the cardinality of the optimal solution) even for intersection graphs of unit disks and axis-aligned unit squares in the plane \cite{Marx05}, making the existence of an efficient polynomial-time approximation scheme (EPTAS) unlikely. However, polynomial-time approximation schemes (PTAS) have been established for fat objects, including squares and disks 
\cite{erlebach2005polynomial, Chan03}.
In the case of $d$-dimensional axis-aligned hyperrectangles ($d\ge3$), the current best approximation ratio stands at $O((\log n)^{d-2} \cdot\log \log n)$ \cite{chalermsook2009maximum}.
For axis-aligned rectangles, Mitchell \cite{mitchell2022approximating} provided the first $O(1)$-approximation algorithm, and subsequently, a $(2+\eps)$-approximation 
was achieved \cite{galvez20211+, galvez20212+}.  
However, addressing the problem for arbitrary line segments, not necessarily axis-parallel, remains significantly challenging, with only a polynomial-time $n^\eps$-approximation currently available \cite{fox2011computing, Cslovjecsek24}. 

These problems have also recently received attention in the online setting, where 
the input is revealed to the algorithm by an adversary in parts; upon receiving a part, the algorithm must make an irrevocable decision for that part while adhering to the constraints of the problem,  before any other parts are revealed. The adversary stops revealing the input at a certain step, at which point the performance of the algorithm can be measured in terms of its {\em competitive ratio}---the ratio between the value obtained by the algorithm and the optimum offline value possible for that input.
Any online algorithm, even if randomized, for interval scheduling faces an $\Omega(n)$ lower bound on its competitive ratio in this adversarial model \cite{halldorsson2002online}. 

Random-order models (ROM) were introduced to mitigate excessively pessimistic results in the adversarial-order arrival model, presenting a more realistic representation in many scenarios.
In this model, the input set of items is chosen by the adversary; however, the arrival order of the items is determined by a permutation selected uniformly at random from
the set of all permutations. 
This random reshuffling of the input items often weakens the adversary, resulting in improved performance guarantees. 
The competitive ratio of an online algorithm in the random-order arrival model, also known as the {\em random-order ratio}, is the ratio between the expected value obtained by the algorithm and the optimum offline value. Moreover, it is often preferable to have a high probability guarantee on the value obtained by the algorithm rather than only a guarantee on its expectation.
ROM encompasses various other commonly used stochastic models, such as the IID model (where the input sequence consists of independent and identically distributed random samples drawn from some fixed but unknown distribution) and is closely related to the prophet model. 
An algorithm in ROM also implies an algorithm under the IID model with, at most, the same competitive ratio. 
For an in-depth discussion, see the survey on these models by Gupta and Singla \cite{Gupta020}. 
Many problems related to scheduling \cite{AlbersJ21, AlbersJ21a, AlbersGJ23, AlbersKL21, AlbersKL21a, gobel2014online, Hebbar24} have received recent attention in this model.

Formally, an online algorithm for the maximum independent set problem in the random-order model receives an input set $\mathcal T$ consisting of $d$-dimensional geometric objects in multiple steps.
Initially, the algorithm is provided with the total number of objects, $n= |\mathcal T|$.\footnote{We will see later that the requirement that the online algorithm knows the size of the instance $n$ in the beginning, is important to obtain non-trivial results for our model (See \cref{prop:sizeObliviousLowerBound} in \cref{sec:proofsOfProps}).} 
Subsequently, the objects from $\mathcal T$ are presented to the algorithm in a uniformly random order, one at a time.
Upon receiving an object from $\mathcal T$,
the algorithm must make an irrevocable decision on whether or not to include the given object in the solution before receiving any further objects. The algorithm's objective is to maximize the size of the solution subject to the constraint that the selected objects are pairwise disjoint. 
We say an algorithm, $\alg$, is $c$-competitive ($c>1$) in the random-order model if, for all inputs $\mathcal T$, 
$\mathbb E[|\alg(\mathcal T)|] \geq \frac 1 c\cdot \opt(\mathcal T) - k$. 
Here, $\opt(\mathcal T
)$ represents the size of a maximum cardinality subset of $\mathcal T$ consisting of disjoint objects, and $k$ is some constant independent of $|\mathcal T|$. The expectation is taken over the randomness in the arrival order and any coin tosses made by the algorithm. 
We say an algorithm, $\alg$, is {\em strictly} $c$-competitive ($c>1$) in the random-order model (also called the absolute random-order ratio in \cite{AlbersKL21a}) if, for all inputs $\mathcal T$, 
$\mathbb E[|\alg(\mathcal T)|] \geq \frac 1 c\cdot \opt(\mathcal T)$.
In contrast to offline randomized algorithms, where repeated runs amplify the success probability, in the online setting, expectation guarantees do not readily translate into high-probability guarantees.
Such high probability guarantees are often needed and thus studied for online algorithms (e.g., \cite{leonardi2001line, KommKKM22, bhattacharya2021online, Mihail21}). 
In this work, the study of such a stronger notion in the context of the random-order arrival model is initiated. 
We say that $\alg$ is {\it strongly} $c$-competitive if for all inputs $\mathcal T$,
$\pr[|\alg(\mathcal T)| \geq \frac 1 c\cdot\opt(\mathcal T)] = 1 - o(1)$,
where the $o(1)$ is with respect to $|\mathcal T|$. Under this stronger notion of competitiveness, it is evident that numerous classical problems and algorithms remain to be explored.

However, even in ROM, obtaining an $n^{o(1)}$-competitive algorithm for interval scheduling is nontrivial. In fact, the trivial greedy algorithm has a competitive ratio $n^{\Omega(1)}$ (see~\Cref{thm:greedy lower}). 
G{\"o}bel et al.~\cite{gobel2014online} studied the problem on graphs with bounded inductive independence number\footnote{The {\em inductive independence number} $\rho$ of a graph is the smallest number for which there exists an order $\prec$ such that for any
independent set $U \subseteq V$ and any $v \in V$, we have $|\{u \in U |\; u \succ v \text{ and } \{u, v\} \in E\}| \leq \rho$.} $\rho$, and designed an online algorithm with a competitive ratio $O(\rho^2)$, {\em in expectation}. 
This already implies $O(1)$-competitiveness for interval scheduling.
However, they did not provide any high-probability guarantee. In fact, no online algorithm for interval scheduling that has an initial {\it observation phase} where it does not pick any of the initial $c n$ intervals (for some absolute constant $c\in (0,1)$) can be strongly $O(1)$-competitive (\cref{prop:observationLowerBound} in \cref{sec:proofsOfProps}). In particular, their algorithm
is not strongly $O(1)$-competitive for interval scheduling.
Additionally, their result does not yield a competitive algorithm for rectangles, where the inductive independence number can be $\Omega(n)$ (\cref{thm:inducindep} in \cref{sec:proofsOfProps}).

Recently, Henzinger et al.~\cite{Henzinger0W20} studied fully dynamic approximation algorithms for {\sc Maximum Independent Set of Hyperrectangles}, where 
each update involves the insertion or deletion of a hyperrectangle. They presented a dynamic algorithm with polylogarithmic update time and a competitive ratio of $(1 + \eps) \log^{d-1} K$, assuming that the coordinates of all input hyperrectangles are in $[0, K]^d$, and each of their edges has length at least 1. It is worth noting that, unlike dynamic algorithms, in online algorithms, decisions are irrevocable, and once hyperrrectangles are selected, they cannot be removed.

\subsection{Our Contributions}
In this work, we present a unified approach to design and analyze online algorithms that achieve {\em strong competitiveness} (high probability performance guarantees in the random-order arrival model). We elucidate our approach by studying the maximum independent set problem across a variety of geometric objects, while the techniques introduced are possibly applicable to a variety of combinatorial problems.
These algorithms are remarkably simple, and we demonstrate that, with the use of appropriate data structures, they operate in near-linear time. 
We obtain strongly $(\log n)^{O(1)}$-competitive algorithms for intervals, rectangles, $d$-dimensional hyperrectangles, fat objects, ellipsoids, and more. 

First, we consider the case of intervals with integral endpoints lying in $[0,K]$  for some positive integer $K=n^{O(1)}$. Note that $K$, in general, can be quite large, e.g., $2^{\Theta(n)}$ as even then the input representation requires only bit size that is polynomial in $n$. Later, we will see how to get rid of this assumption on the value of $K$.
We observe that, for intervals, if the intervals are of similar lengths, 
 or if the instance is sparse (not too many intervals lie within any unit range,
 i.e., the underlying graph has a small maximum degree), then 
the greedy algorithm is $O(1)$-competitive.
To leverage this observation, we categorize the intervals into $O(\log n)$ classes, ensuring that intervals in the same class have roughly similar lengths. The aim is to run the greedy algorithm on one of these classes, which has an independent set of size at least $\opt/O(\log n)$.
Our algorithm begins with an initial \emph{observation phase}, lasting for $n/2$ steps, at the end of which the algorithm picks a class that had the largest number of disjoint intervals. Subsequently, in the \emph{action phase}, our algorithm greedily selects intervals arriving only from the identified class. 
With a loss of an additional $O(\log\log n)$-factor, we ensure that the optimum value in classes with large independent sets is evenly concentrated in both phases, with high probability, thus fulfilling our aim.
The algorithm can be easily implemented in $O(n\log n)$ time using binary search trees.
This leads to the following theorem.

\begin{theorem}\label{thm:interBounded}
 There exists a strongly $O(\log n \cdot \log\log n)$-competitive ${O}(n\log n)$-time algorithm for interval scheduling in ROM, where the intervals have integral endpoints in $[0,K]$ and $K=n^{O(1)}$.
\end{theorem}

In fact, we demonstrate that our results hold for the general case of $(K,D)$-bounded instances (see \Cref{sec:Intervals} for the definition). Intuitively, this implies that the instance lies in $[0,K]$, and each unit range contains at most $D$ intervals. We observe that for intervals with lengths at most $1$, the greedy algorithm is $O(D)$-competitive. The observation becomes crucial for deciding which class to choose at the end of the observation phase; there is a trade-off between the optimal value available in a class and the loss in the competitive ratio that the greedy algorithm will incur on that class.
Such instances later play an important role in removing the assumption on $K$. 

Next, we show that our techniques can be generalized to hyperrectangles. One crucial difference, however, is that we cannot solve the independent set problem exactly for the various classes in the observation phase. Nevertheless, we show that we can still employ the greedy algorithm for estimating the optimal values for these classes. While a trivial implementation of the algorithm would take $\tilde{O}(n^2)$ time, the use of novel data structures 
allows us to improve the runtime to $O_d(n\log n)$ (see \Cref{sec:datastructures}).\footnote{$O_d(f(n))$ refers to a function in $O(f(n))$ when $d$ is a constant.} Intuitively, for each class, we maintain a uniform $d$-dimensional grid, where each grid cell stores all the hyperrectangles intersecting that cell. The grid cells are spaced in such a way that each cell is guaranteed to intersect only a constant number of hyperrectangles, and further, each hyperrectangle intersects only a constant number of grid cells. This ensures that the greedy algorithm only needs to examine $O(1)$ hyperrectangles (instead of $O(n)$) in order to decide whether or not to select an incoming hyperrectangle. We show that $d$ binary search operations (one for each dimension) suffice to enumerate the aforesaid hyperrectangles, implying a total execution time of $O_d(n\log n)$ for our algorithm.

\begin{theorem}\label{thm:HyperrecBounded}
For all constants $d\geq 2$, there exists a strongly $(\log n)^{O(d)}$-competitive  ${O}_d(n\log n)$-time algorithm for {\sc Maximum Independent Set of $d$-dimensional Hyperrectangles} in ROM when the hyperrectangles have integral endpoints in $[0,K]^d$ and $K=n^{(\log n)^{O(1)}}$.
\end{theorem}

Furthermore, we extend the results to $\sigma$-rectangular objects (see \Cref{defn:fat}). Intuitively, these objects can be made fat after some uniform scaling along the axes, and they encompass many important objects such as axis-aligned hyperrectangles, hyperspheres, fat objects, ellipsoids, etc.

\begin{theorem}\label{thm:fatBounded}
For all constants $d\geq 2$, there exists a strongly $(\sigma\log K)^{O(d)}$-competitive  $O_d(n\log n)$-time algorithm for the maximum independent set problem in ROM for $\sigma$-rectangular objects in $[0,K]^d$ and $K=n^{(\log n)^{O(1)}}$.
\end{theorem}

This provides an intriguing characterization of the problem; for thin 
objects, including arbitrary line segments, there is no known polynomial-time algorithm with a polylogarithmic approximation ratio, even in the offline setting. 

Our bounds on competitive ratios above are indeed $(\log K)^{O(d)}$.
Therefore, we assume $K$ to be quasi-polynomially bounded to obtain a  $(\log n)^{O(d)}$ guarantee. This assumption on the bounding box is common in many geometric problems, such as geometric knapsack~\cite{AdamaszekW15}, storage allocation problem~\cite{MomkeW20}, etc. 
For instance, the result of Henzinger et al.~\cite{Henzinger0W20} for hyperrectangles in the dynamic setting holds under this assumption.
Surprisingly, we are able to eliminate the assumption on $K$. For this, our algorithm has an initial {\em scale preparation} phase lasting for about $n/2$ steps to learn about the instances and project $[0, K]$ to $[0, n]$. Using hypergeometric concentration inequalities, we show that our projection creates a $(K,D)$-bounded instance with $K=O(n)$ and $D=O(\log n)$. 

Roughly speaking, the key observation that facilitates such a projection is as follows. 
Let $L$ denote the set of all left endpoints of the intervals in the input instance and $\hat L\subseteq L$ denote the set of those left endpoints observed during the scale preparation phase. Then, $\hat L$ forms a random subset of $L$ of size $n/2$ selected uniformly at random. Consequently, the left endpoints in $\hat L$ are evenly distributed among the left endpoints in $L$. In particular, between any two consecutive elements of $\hat L$ (when it is sorted), at most $O(\log n)$ elements of $L\setminus \hat L$ appear with high probability. In the subsequent non-uniform scale that we prepare, the left endpoints of $\hat L$ are mapped to $\{1,2,\dots,n/2\}$, while the remaining points on the real line are projected appropriately, maintaining their original ordering. When the intervals are read in this non-uniform scale, at most $O(\log n)$ intervals are contained within $(i,i+1)$ for all $i<n/2$ with high probability. Thus, the original input instance when viewed under such a non-uniform scale appears to be a $(K,D)$-bounded instance for $K=O(n)$ and $D=O(\log n)$.

Our technique of mapping $[0,K]$ to $[0, n]$ might prove valuable in other related problems (e.g., problems related to geometric intersection graphs of objects confined within a bounding box of arbitrary size).
Additionally, we crucially leverage the fact that axis-aligned rectangles remain axis-aligned rectangles under non-uniform scaling along the axes. 
We then apply our previous algorithms as subroutines and obtain the following theorems.
 
\begin{restatable}{theorem}{interUnbound}
\label{thm:interUnbound}
 There exists a strongly $O(\log n \cdot \log\log n)$-competitive ${O}(n\log n)$-time algorithm for interval scheduling in ROM.
\end{restatable}

\begin{restatable}{theorem}{hyperUnbound}
\label{thm:hyperUnbound}
 For all constants $d\geq 2$, there exists a strongly $O((\log n)^d\cdot \log\log n)$-competitive ${O}_d(n\log n)$-time algorithm for {\sc Maximum Independent Set of $d$-dimensional Hyperrectangles} in ROM.
\end{restatable}

Note that we only consider the unweighted setting. High probability guarantees are not possible in the weighted setting as even in the special case of the {\em Secretary problem} (when all intervals have the same start and endpoints, but may have different weights), the probability of success 
is at most $1/e$ \cite{dynkin1963optimal}.

In summary, we introduce a unified framework designed to achieve strongly competitive online algorithms in the random-order model when dealing with geometric objects confined within a bounded box (which may be of arbitrary size). As any permutation of the objects is equally likely, the underlying structure and arrangement of the objects are effectively captured within a constant fraction of the input. Consequently, with a negligible $O(1)$ loss in the competitive ratio, the size of the bounding box can be reduced to $K=n^{O(1)}$. Following this reduction, the input is carefully decomposed into $k=(\log n)^{O(1)}$ classes, each class $i$ admitting a nice $\alpha_i$-competitive algorithm. The effective balance of these $\alpha_i$'s results in an overall competitive ratio smaller than $\tilde O(k \cdot \max \alpha_i)$ for the original problem. For different problems, the ingenuity lies in devising a clever decomposition into classes, ensuring effective approximation of the classes and a balanced distribution of the approximation factors. 
Our algorithms for {\sc Maximum Independent Set of Hyperrectangles} are not only simple and fast but also nearly match the best offline approximation guarantees, falling short by only a $\log^2 n$ factor.

\subsection{Further Related Work}
The study of online interval scheduling was initiated by Lipton and Tomkins \cite{lipton1994online}, who presented an $O((\log \Delta)^{1+\epsilon})$-competitive algorithm when intervals arrive in the order of their start times, where $\Delta$ is the ratio of the lengths of the longest and shortest intervals in the input. However, due to the
$\Omega(n)$ lower bound in the adversarial-order arrival model, online interval scheduling has only been explored in restricted settings \cite{woeginger1994line, fung2008online, seiden1998randomized, zheng2013online, borodin, 10.1007/978-3-031-38906-1_14, BACHMANN20131, yu2020primal}.

Recently, De, Khurana, and Singh~\cite{de2021online} studied the online independent set and online dominating set problems, presenting an algorithm with a competitive ratio equal to the independent kissing number of the corresponding geometric intersection graph.\footnote{For a family of geometric objects $\SS$, the independent kissing number is defined as the maximum number of pairwise non-touching objects in $\SS$ that can be arranged in and around a particular object such that all of them are dominated/intersected by that object.} However, these results do not provide meaningful bounds for intervals or hyperrectangles, as considered in our work. 
Building on the work of Henzinger et al.~\cite{Henzinger0W20}, Bhore et al.~\cite{bhore2020dynamic} investigated the maximum independent set problem in the fully dynamic setting for various classes of geometric objects, including intervals and $d$-dimensional hypercubes. They showed that a $(1+\eps)$-approximation can be maintained in logarithmic update time for intervals, which generalizes to an $O(4^d)$-competitive algorithm with polylogarithmic update time for hypercubes.

The {\sc Maximum Independent Set of Hyperrectangles} is closely related to many fundamental problems in geometric approximation, such as geometric knapsack \cite{lGalvezGIHKW21, Jansen0LS22, GalvezRW21}, geometric bin packing \cite{BansalK14, KhanS23},  strip packing \cite{GalvezGAJKR23, KhanLMSW22}, etc. For more results on multidimensional packing and covering problems, see~\cite{ChristensenKPT17}. Various geometric problems have been extensively studied in the online setting, including geometric set cover \cite{KhanLRSW23}, hitting set \cite{even2014hitting}, piercing set \cite{de2022online} etc. The streaming model has also received some attention~\cite{cabello2017interval,bhore2022streaming, czumaj2022streaming}. Additionally, the random-order arrival model has been explored for classical problems such as 
knapsack \cite{AlbersKL21, KesselheimRTV18}, bin packing \cite{kenyon1995best, AlbersKL21a, AyyadevaraD0S22, Hebbar24}, set cover \cite{gupta2021random}, matroid intersection \cite{guruganesh2017online}, matching \cite{fu2021random}, and more. 

\textbf{Organization of the paper.} In \Cref{sec:prelim}, we discuss preliminaries. 
 \Cref{sec:Intervals} deals with intervals contained within a bounding box of size polynomial in $n$ and we prove \Cref{thm:interBounded}. We then generalize this result to hyperrectangles and other objects in  \Cref{sec:HyperrectGeneralize}, proving Theorems~\ref{thm:HyperrecBounded} and \ref{thm:fatBounded}. \Cref{sec:KtofnMap} contains the proofs for intervals and hyperrectangles confined within arbitrary sized bounding box:  Theorems~\ref{thm:interUnbound} and \ref{thm:hyperUnbound}.   Finally, \cref{sec:conclusion} contains some concluding remarks. 
 Due to space limitations, some parts have been moved to the appendix, including \Cref{prop:sizeObliviousLowerBound,thm:inducindep,prop:observationLowerBound}. Also, \cref{sec:datastructures} contains a discussion on how to efficiently implement our algorithms using appropriate data structures.
%%%%%%%End of Introduction%%%%

%%%%%%%%%%Preliminaries%%%%%%%
\section{Preliminaries}
\label{sec:prelim}
For all positive integers $n$, let $[n]:=\{ 1, 2, \dots, n\}$. 
We will assume the base of $\log$ is $e$, unless otherwise specified. For sets $A$ and $B$, $A\uplus B$ denotes their disjoint union.
For an interval $I$, let $||I||$ denote its length.
For the case of hyperrectangles, we will assume $\HH: = \{ H_1, H_2, \dots, H_n\}$ to be the input set of $n$ $d$-dimensional hyperectangles, where $H_i:=[x_{i,1}^{1}, x_{i,1}^{2}] \times [x_{i,2}^{1}, x_{i,2}^{2}]\times \dots \times [x_{i,d}^{1}, x_{i,d}^{2}]$. Here, $0\le x_{i,j}^1 \le x_{i,j}^2  \le K$ for all $i\in[n], j \in[d]$, and the length of $H_i$ along dimension $j$ is $l_j(H_i):= x_{i,j}^{2}- x_{i,j}^{1}$.
Two objects $H_i$  and $H_j$ {\em intersect} if $H_i \cap H_j\neq\emptyset$. Let $\OPT(\HH)$ denote a maximum independent set for the intersection graph corresponding to $\HH$. The {\sc Maximum Independent Set of Hyperrectangles} problem aims to determine $\OPT(\HH)$. Let $\opt(\HH) := |\OPT(\HH)|$ be the size of
an optimal solution.

The following lemma will be crucial in our analysis (see \Cref{sec:hoeffding} for a proof).

\begin{restatable}[Hypergeometric concentration]{lemma}{hoeffding}
\label{lem:hoeffding}
There are $N$ balls of which $M$ are red. Some $n$ balls are sampled from the $N$ balls uniformly at random without replacement. Let $X$ be the number of red balls that appear in the sample. Let $p=\frac M N$ and $0\leq \delta \leq 1$. Then,
\begin{align*}
\pr[X \geq (1+\delta) pn]\leq \exp(-\delta^2 pn /3); \
\pr[X \leq (1-\delta) pn]\leq \exp(-\delta^2 pn /2).
\end{align*}
\end{restatable}

\subsection{Greedy Algorithm}
Throughout the paper, we let $\greedy$ represent the greedy algorithm, which selects an object on arrival if it does not intersect previously selected objects. Clearly, $\greedy$ returns a maximal independent set of the geometric intersection graph. We first show that the competitive ratio of $\greedy$ can be $\Omega(\sqrt{n})$ in expectation on arbitrary inputs.

\begin{figure}[h]
    \centering
     \includegraphics[width=0.6\linewidth]{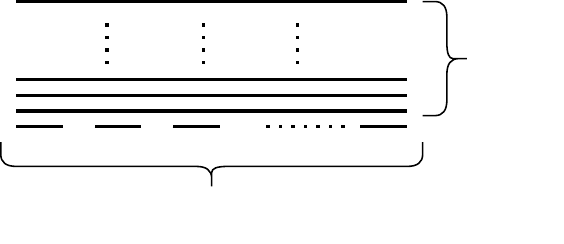}
    \caption{Instance $\II$}
    \label{fig:LBinstance}
\end{figure}
\begin{theorem}
\label{thm:greedy lower}
    $\greedy$ has a competitive ratio of at least $\Omega(\sqrt{n})$ in the random-order model.
\end{theorem}
\begin{proof}
    Consider the input instance $\II:= \II_1\cup \II_2$, where $\II_1$ consists of $\sqrt{n}$ disjoint unit-length intervals, and $\II_2$ comprises $n-\sqrt{n}$ identical intervals each overlapping all intervals in $\II_1$ (see \Cref{fig:LBinstance}). Clearly, the optimal solution consists of all intervals in $\II_1$. Now when the input $\II$ is presented in a random order, we have two possible cases to consider.
    
    \textit{Case 1:} The first interval to arrive belongs to $\II_1$. This happens with probability $\sqrt{n}/n=1/\sqrt{n}$. In this case, $\greedy$ selects this interval and all future intervals from $\II_1$ as and when they arrive, but no interval from $\II_2$ is selected. 

    \textit{Case 2:} The first interval to arrive is from $\II_2$. In this case, $\greedy$ selects this interval and therefore no further interval is chosen in any future iteration.

    Hence, the expected number of intervals selected by $\greedy$ is upper bounded by $(1/\sqrt{n})\cdot |\II_1|+1=2$. Since $\opt(\II)=|\II_1|=\sqrt{n}$, we have a lower bound of $\Omega(\sqrt{n})$ as claimed.
\end{proof}

Next, we state two properties of $\greedy$ that we shall frequently use. The first property states that $\greedy$ performs well when the input instance consists of hyperrectangles of similar size.

\begin{figure}[h]
    \centering
    \includegraphics[width=0.4\linewidth]{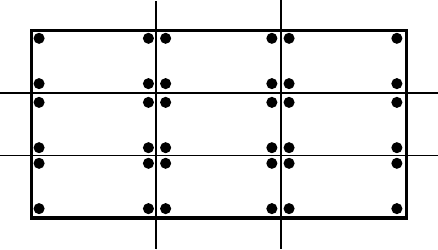}
    \caption{For $d=2$, $\Delta=3$, a piercing set of $\HH$ can be formed by placing 36 points inside each $\HH'$}
    \label{fig:piercingset}
\end{figure}

\begin{restatable}{lemma}{greedyBounded}
\label{lem:greedyBounded}
Let $\HH$ be a set of hyperrectangles such that for any $H_1, H_2\in \HH$ and for any dimension $j\in [d]$, we have $l_j(H_1)/l_j(H_2) \in [1/\Delta, \Delta]$ for some integer $\Delta$. Let $\HH'$ be any maximal independent subset of hyperrectangles from $\HH$. then $\opt(\HH) \le (2\Delta)^d \cdot |\HH'|$.
\end{restatable}
\begin{proof}
    Let $P^*$ be the minimum piercing set of $\HH$, i.e., a set of points of minimum size such that each hyperrectangle in $\HH$ is stabbed by at least one point in $P^*$. Then clearly $\opt(\HH)\le |P^*|$, since no two hyperrectangles in $\opt(\HH)$ can be stabbed by the same point in $P^*$. We next show that $|P^*|\le (2\Delta)^d \cdot |\HH'|$, thus completing the proof. Consider any $H\in \HH'$. We divide $H$ into $\Delta^d$ identical homothetic copies using $\Delta-1$ equally-spaced axis-aligned hyperplanes in each dimension (see \Cref{fig:piercingset}). Then observe that any hyperrectangle that intersects $H$ must contain a corner of at least one of these smaller hyperrectangles because of the assumption on the ratio of the side lengths. Since $\HH'$ is maximal, it follows that the set of $(2\Delta)^d$ corners of these $\Delta^d$ hyperrectangles form a piercing set of the hyperrectangles intersecting $H$. Hence $|P^*|\le (2\Delta)^d \cdot |\HH'|$. 
\end{proof}

The next lemma is based on the fact that $\greedy$ performs satisfactorily when the maximum degree of the underlying intersection graph is small. 

\begin{restatable}{lemma}{greedyCompetitive}
\label{lem:greedyCompetitiveRatioprelim}
    Let $\HH$ be a set of hyperrectangles such that $l_j(H_i)\le 1$ for all $H_i \in \HH$ in some dimension $j \in [d]$. Also suppose for each integer $m\in [K]$, the number of hyperrectangles $H_i$ with $x_{i,j}^1 \in (m-1,m)$ is at most $D$, where $D \in \mathbb{N}$. Then $\opt(\HH) \le 3D \cdot |\greedy(\HH)|$.
\end{restatable}
\begin{proof}
    The conditions of the lemma imply that each hyperrectangle in $\HH$ can overlap with at most $3D-1$ other hyperrectangles. Therefore, the maximum degree of a node in the intersection graph of $\HH$ is at most $3D-1$. It is a folklore result that the greedy algorithm selects an independent set of size at least $n/(\Delta +1)$ in any $n$-vertex graph with maximum degree $\Delta$. Therefore $\opt(\HH) \le 3D \cdot |\greedy(\HH)|$. 
\end{proof}
%%%%%End of Preliminaries%%%%%

%%%%%%%%%Intervals%%%%%%%%%%%%
\section{Online Algorithms For Intervals}
\label{sec:Intervals}
In this section, we prove \Cref{thm:interBounded}. 
We first prove the following general lemma for \textit{bounded} instances, which are defined as follows.

\begin{definition}
   A set of intervals $\mathcal I$ is {\em $(K,D)$-bounded} if each $I\in \mathcal I$ is a subset of $[0,K]$ and for all $i \in\mathbb{Z}$, the number of intervals from $\mathcal I$ that are a subset of $(i, i + 1)$ is at most $D$.
\end{definition}

We show that there exists an online algorithm which on $(K,D)$-bounded instances is $O(D + \log K\cdot \log\log K)$-competitive with probability $1-o_K(1)$.

\begin{lemma}\label{lem:bounded}
    There exists absolute constants $\lambda \in(0,1)$, $\mu \geq 1$, and an online algorithm such that for each $(K,D)$-bounded instance $\mathcal I$, where $K$ and $n:=|\mathcal I|$ are large enough, the algorithm outputs an independent set from $\mathcal I$ of size at least $\frac \lambda {D+\log K\cdot \log\log K}\cdot\opt(\mathcal I)$ with probability at least $1-\frac{\mu}{ \log K}$.
\end{lemma}

\begin{proof}
Let $\eps:=1$  
and $k:= \lceil \log_{1+\eps} K\rceil$. 
We partition $\mathcal I$ into $k+1$ sets: $S_0, \dots, S_k$, depending on the lengths of the intervals in $\mathcal I$, such that $S_0:=\{ I \in \mathcal I \mid ||I|| \in [0,1] \}$ and $S_i := \{ I\in \mathcal I \mid ||I|| \in ((1+\eps)^{i-1},(1+\eps)^i]\}$ for all $i \in [k]$. We now describe our algorithm before moving to its analysis.

\textbf{Algorithm.} 
     The algorithm has two phases: the \textit{observation} phase, when it receives the first $\lceil \frac n 2 \rceil$ intervals (we call them $\mathcal I_0$) and the \textit{action} phase, when it receives the remaining $n-\lceil\frac n 2 \rceil$ intervals (we call them $\mathcal I_1$).
    Note that $\mathcal I=\mathcal I_0 \uplus \mathcal I_1$, where $\mathcal I_0$ is a uniformly random subset of $\mathcal I$ of size $\lceil\frac n 2\rceil$, and $\mathcal I_1$ is a uniformly random subset of $\mathcal I$ of size $n-\lceil\frac n 2\rceil$.
    \begin{itemize}
    \item{\textbf{Observation phase}} In this phase intervals in $\mathcal I_0$ arrive.
    \begin{itemize}
    \item{Initialization:} Initialize $L_i=\emptyset$ for $i\in\{0\}\cup[k]$. 
    \item When a new interval $I\in \mathcal I_0$ arrives, add $I$ to $L_i$ iff $I\in S_i$.
    \item At the end of the observation phase, all $L_i$'s are such that  
    $L_i = \mathcal I_0 \cap S_i.$ 
    Compute $\opt(L_i)$ for each $i\in \{0\}\cup [k]$. 
    \item
   Compute the index $m\in \{0\}\cup[k]$ as follows. Let $m\in[k]$ be an index such that $\opt(L_m) \geq \opt(L_i)$ for all $i\in [k]$. If $\opt(L_0) > k \cdot \opt(L_m)$, set $m=0$.
    \end{itemize}
    \item{\textbf{Action phase}} In this phase, the intervals in $\mathcal I_1$ arrive. 
    \begin{itemize}
        \item{Initialization:} $R_m=\emptyset$. On arrival of an interval $I\in \mathcal I_1$, do the following.
    \item If $I\in S_m$ and $I$ does not intersect any intervals in $R_m$ select $I$ and add $I$ to $R_m$.
    \item Else, if
    $I$ is the last interval to arrive and $R_m=\emptyset$, then select $I$.
    \item Else, if $I\notin S_m$, then do not select $I$.
    \end{itemize}
    \end{itemize}
 Thus, in the observation phase, the algorithm computes $\opt(L_i)$, where $L_i = S_i\cap \mathcal I_0$, for all $i\in {0}\cup[k]$. Then, 
 it sets $m=0$ when $\opt(L_0)>k\cdot\opt(L_i)$ for all $i\in [k]$, else it  sets $m\in[k]$ such that $\opt(L_m)\geq \opt(L_i)$ for all $i\in [k]$. In the action phase, the algorithm simply runs the greedy algorithm $\greedy$ on $R_m := \mathcal I_1 \cap S_m$, ignoring all intervals not in $S_m$. When the algorithm receives the very last interval and $R_m = \emptyset$, it picks the last interval. This ensures that the algorithm outputs an independent set of size at least one whenever $n
    \geq 2$.

\textbf{Analysis.} We now analyze the competitive ratio of the above algorithm.
\begin{align}
\text{We may assume that }\opt(\mathcal I) \geq 10^4 k\log k. \label{eq:assumption}
\end{align} Otherwise, since the algorithm is guaranteed to pick at least one interval (for $n\geq 2$), we get a  competitive ratio of $10^4k \log k$.

We now define an event $E$ and show that (I) if $E$ holds, then the algorithm achieves a competitive ratio of $O(D+K)$, and (II) $E$ holds with probability $1-o_k(1)$.
In order to define the event $E$, we first define $L_i = \mathcal I_0 \cap S_i$ and $R_i = \mathcal I_1 \cap S_i$ for all $i\in \{0\}\cup [k]$, as in the algorithm above. Furthermore, let us define $B=\{i\in \{0\}\cup [k] \mid \opt(S_i) \geq 10^3\log k\}$.

We say that the event $E$ holds iff for all $i\in B$, \begin{align}\min(\opt(L_i),\opt(R_i))\geq (1-\delta)\frac {\opt(S_i)} 2,\quad\text{where $\delta = \frac 1{10}.$}\label{eq1} \end{align}

\noindent\textbf{(I) Assuming $E$ holds, the algorithm is $O(D+k)$-competitive.}

We prove two basic bounds that will help us in the proof. First, we want to show that $\sum_{i\in B} \opt(S_i)$ is large, i.e., at least a constant fraction of $\opt(\mathcal I)$. Intuitively, this would imply that even if an algorithm ignored intervals in $\cup_{i\notin B} S_i$, it would have substantial value available.
\begin{align}
     \opt(\mathcal I)\leq \sum_{i\in B} \opt(S_i) + \sum_{i\notin B} \opt(S_i) 
    \implies \frac{\opt(\mathcal I)}2 \leq\sum_{i\in B} \opt(S_i),\label{eq2}
\end{align}

where the last inequality holds as $\sum_{i\notin B} \opt(S_i) \leq (k+1)\cdot10^3\log k\leq 5\cdot10^3k\log k\leq \frac{\opt(\mathcal I)}2$.

Next, we want to show for $i\in B$, $\opt(R_i)$ is at least a constant fraction of $\opt(L_i)$. 
\begin{align}\text{For each } i\in B, \ \opt(R_i) \geq \frac{1-\delta}2\opt(S_i)\geq \frac{1-\delta}2 \opt(L_i),\label{eq3}\end{align}
where the last inequality holds as $L_i \subseteq S_i$. We now prove the following claim.

\begin{claim}
    Let $m$ be the index chosen at the end of the observation phase, then $m\in B$, 
\end{claim}
\begin{proof}[Proof of claim.]
  We assume $m\notin B$ and derive a contradiction. From Bounds \eqref{eq1} and \eqref{eq2}:

\begin{align}
\sum_{i\in B}\opt(L_i)\ge \frac{1-\delta}{4}\opt(\II)\ge \frac{(1-\delta)\cdot 10^4}{4} k\log k=\frac{9}{4} \cdot 10^3\cdot k\log k.\label{inequality:li2}
\end{align}
Whereas, if $m\not\in B$, then 
\begin{align}
    \sum_{i\in B}\opt(L_i)\leq\sum_{i=0}^k \opt(L_i) \leq k\cdot 10^3\log k + k\cdot 10^3\log k = 2\cdot 10^3\cdot k\log k\label{inequality:li} 
\end{align}
The last inequality follows from the facts that  for all $i\in[k]$, $\opt(L_i)\leq\opt(L_m)$, $\opt(L_0)\leq k\cdot\opt(L_m)$, and as $m\notin B$ we have $\opt(L_m) \leq \opt(S_m)\leq 10^3\log k$. Thus, from inequalities \eqref{inequality:li2} and \eqref{inequality:li}, we get a contradiction. Thus, $m\in B$.
\end{proof}

We now have two cases depending on whether the algorithm picks $m = 0$ or $m \neq 0$ at the end of the observation phase.
In either case, we show $|\greedy(R_m)|\geq \frac 1{O(D+k)}\opt(\mathcal I)$.

\noindent\textbf{Case: $m=0$.}
\begin{align*}
    \frac{\opt(\mathcal I)}2&\leq \sum_{i\in B} \opt(S_i)
    \leq  \frac 2{1-\delta} \sum_{i\in B}\opt(L_i)
    \\
    &\leq \frac 2{1-\delta}\left(\opt(L_0) + \sum_{i\in B\setminus\{0\}}\opt(L_i)\right)\\
    &\leq \frac 2{1-\delta}\left(\opt(L_0) + k\cdot \frac 1 k \cdot \opt(L_0)\right),
\end{align*}
where the first and second inequalities follow from Bounds~\eqref{eq2}, \eqref{eq1}, the last inequality follows from the fact that the algorithm picks $m=0$ when $\opt(L_0) \geq k\cdot\opt(L_i)$ for all $i\in[k]$. Thus,
\begin{align*}
    \opt(L_0) &\geq \frac{1-\delta}8 \opt(\mathcal I)
    \implies \opt(R_0) \geq \frac{(1-\delta)^2}{16}\opt(\mathcal I)\quad\text{(from Bound~\eqref{eq3}).}
\end{align*}

\begin{align}\text{Now from Lemma~\ref{lem:greedyCompetitiveRatioprelim}, the output size }
|\greedy(R_0)|\geq \frac 1{3D} \opt(R_0)\geq \frac{(1-\delta)^2}{48D}\opt(\mathcal I). \label{eq:first}
\end{align}

\noindent\textbf{Case: $m\neq 0$.}
\begin{align*}
    \frac{\opt(\mathcal I)}2
    &\leq \frac 2{1-\delta}\left(\opt(L_0) + \sum_{i\in B\setminus\{0\}}\opt(L_i)\right)
    \leq \frac 2{1-\delta}\left( k \cdot\opt(L_m) + k \cdot \opt(L_m) \right),
\end{align*}
where the first inequality follows the same way as in case $m=0$ above; the last follows from the fact that $\opt(L_0)\leq k\cdot \opt(L_m)$ as $m\neq 0$ and $\opt(L_i)\leq \opt(L_m)$ for all $i\in[k]$.
Thus,
\begin{align*}
    \opt(L_m) &\geq \frac{1-\delta}{8k} \opt(\mathcal I)
    \implies \opt(R_m) \geq \frac{(1-\delta)^2}{16k}\opt(\mathcal I)\quad\text{(from~ Bound \eqref{eq3}).}
\end{align*}

\begin{align}\text{Now from Lemma~\ref{lem:greedyBounded}, the output size }
    |\greedy(R_m)|\geq \frac 1{4} \opt(R_m)\geq \frac{(1-\delta)^2}{64k}\opt(\mathcal I). \label{eq:second}
\end{align}

\begin{align}\text{Thus, in either case, from Bounds \eqref{eq:first} and \eqref{eq:second}, } |\greedy(R_m)|\geq \frac{(1-\delta)^2}{64(D+k)}\opt(\mathcal I).\label{eq:result}\end{align}

Thus, from Bounds \eqref{eq:assumption} and \eqref{eq:result}, our algorithm is $O(\max(k\log k, D+k))$-competitive, which implies that it is $O(D+k\log k)$-competitive.
We now show that (II) is true.

\noindent\textbf{(II) $E$ holds with probability at least $1-3/k$.}

Fix an $i\in B$. 
In the following, we use the fact that $\mathcal I_0$ is a random subset of $\mathcal I$ of size $\lceil\frac n 2\rceil$, and upper bound the probability that only a small fraction of $\OPT(S_i)$ is picked in $\mathcal I_0$.

\begin{align*}
\pr&\left[\opt(L_i) < (1-\delta)\frac{\opt(S_i)} 2\right] \leq  \pr\left[\opt(L_i) \leq (1-\delta)\frac{\opt(S_i)} n \left\lceil \frac n 2 \right\rceil\right] \\
&\leq \exp\left(- \frac {\delta^2}2 \frac{\opt(S_i)}n \left\lceil \frac n 2\right\rceil\right) 
\leq \exp\left(- \frac {\delta^2}4 \opt(S_i)\right)
\leq \exp\left(- \frac {10^3\delta^2}4 \log k\right)
\leq \frac 1 {k^2},
\end{align*}
where the second inequality follows from Lemma~\ref{lem:hoeffding} where $N, M, P,$ and $n$ in the Lemma statement are set to $n, \opt(S_i), \frac{\opt(S_i)}n,$ and $\lceil \frac n 2 \rceil$, respectively, and the fourth inequality follows since $\opt(S_i)\geq 10^3 \log k$, as $i\in B$. 

A similar calculation, where we use the fact that $\mathcal I_1$ is a uniformly random subset of $\mathcal I$ of
\begin{align*}
\text{size $n-\lceil \frac n 2 \rceil$, yields for all $n \geq 100$, }\quad
\pr\left[\opt(R_i) < (1-\delta)\frac{\opt(S_i)} 2\right] 
&\leq \frac 1 {k^2}.
\end{align*}

Now, from the union bound,
$\pr[\exists i\in B \text{ such that }\min(\opt(L_i),\opt(R_i)) < (1-\delta) \frac{\opt(S_i)}2] \leq \frac{2|B|}{k^2} \leq \frac {2k+2}{k^2}\leq \frac 3 k$, for all $k\geq 2$. Thus, $\pr[E] \geq 1 - \frac 3 k$.
This completes the proof.
\end{proof}

Now, when the intervals have integral endpoints in $[0,K]$, where $K=n^{O(1)}$, then $D=0$.
Plugging these values in the statement of \Cref{lem:bounded}, we obtain
\Cref{thm:interBounded}. 
%%%%%End of Intervals%%%%%%%%%

%%%%%%%%Generalization%%%%%%%%
\section{Generalization to Hyperrectangles and Other Objects}
\label{sec:HyperrectGeneralize}
In this section, we first prove \Cref{thm:HyperrecBounded}. We extend our result obtained for intervals to the maximum independent set of $(K,D)$-bounded set of hyperrectangles (\Cref{def:boundedHyperrectangles} in \Cref{sec:boundedExtensionProof}). The input instance $\mathcal H$ consists of hyperrectangles that are in the bounding box $[0,K]^d$, and for each $i\in [K]$ and $j\in [d]$ the number of hyperrectangles in $\mathcal H$ whose projections along the $j^\text{th}$ axis falls in the interval $(j,j+1)$ is at most $D$. We prove the following general lemma.

\begin{restatable}{lemma}{generalBounded}
\label{lem:generalBounded}
    For each $d\in \mathbb N$, there exist constants $\lambda\in(0,1)$, $\mu\geq 1$, and an online algorithm such that for each $(K,D)$-bounded set of $d$-dimensional hyperrectangles $\mathcal H$ where $K$ and $|\mathcal H|$ are large enough, the algorithm outputs an independent set from $\mathcal H$ of size at least $\frac \lambda {D^2+(\log K)^d\cdot \log\log K}\cdot\opt(\mathcal H)$ with probability at least $1-\frac{\mu}{ \log K}$.
\end{restatable}

The full proof of~\Cref{lem:generalBounded} is in~\Cref{sec:boundedExtensionProof}. The proof is similar to that of~\cref{lem:bounded}, with the algorithm having observation/action phases, with three notable differences:

\begin{itemize}
\item{\textbf{Partitioning:}} We partition the input instance $\mathcal H$ into $d$ + $(\log K)^d$ classes: $S_x$'s for $x\in [d]$ and $S_y$'s for $y\in [\log K]^d$, based on the side lengths of the hyperrectangles in $\mathcal H$. Each $S_x$ has hyperrectangles $H$ such that $l_x(H) \leq 1$
($\greedy$ is $O(D)$-competitive for such instances). Each $S_y$ has hyperrectangles that have similar lengths in each dimension ($\greedy$ is $4^d$-competitive for such instances). Like in the case of intervals, $S_i=L_i\uplus R_i$, where $L_i$ are hyperrectangles in $S_i$ that arrive during the observation phase.
\item{\textbf{Estimating:}} Calculating $\opt(L_i)$ is hard; computing a maximum independent set is NP-complete even for unit squares \cite{fowler1981optimal}. Thus, we estimate their sizes with $\greedy$: $\hat L_i = \greedy(L_i)$.
\item{\textbf{Balancing:}} Similar to the interval case, we shall select only a particular size class $m$ based on the estimates of $\opt$ values seen during the observation phase. Due to the difference in the competitiveness achieved by $\greedy$ on $S_x$'s and $S_y$'s, we need to balance the estimates while choosing $m$. Let $m_1 = \argmax_x \hat L_x$ and $m_2 = \argmax_y \hat L_y$.
  If $\hat L_{m_1} \geq \frac{(k+1)^d}{D} \hat L_{m_2}$ set $m=m_1$, else set $m=m_2$. As before the algorithm outputs $\greedy(R_m)$.
\end{itemize}

For the special case of $K=n^{(\log n)^{O(1)}}$ and the input hyperrectangles having integer coordinates, i.e., $D=0$, \Cref{lem:generalBounded} implies \Cref{thm:HyperrecBounded}.

\subsection{Extension to Ellipses and Fat Objects} 
In this section, we prove \Cref{thm:fatBounded} for $\sigma$-rectangular objects.
First, we define $\sigma$-rectangular objects. 
\begin{definition}
\label{defn:fat}
    An object $F$ is said to be $\sigma$-rectangular $(\sigma>1)$ if there exist axis-aligned inscribing and circumscribing hyperrectangles of $F$: $\In(F)$ and $\Out(F)$, respectively, such that for all $j\in[d]$, $l_j(\In(F))\geq  l_j(\Out(F))/\sigma$.     
\end{definition}

\begin{figure}
\centering
\subfloat[]{\includegraphics[width=0.25\textwidth]{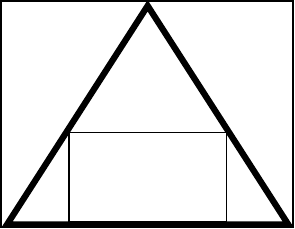}}\qquad
\subfloat[]{\includegraphics[width=0.25\textwidth]{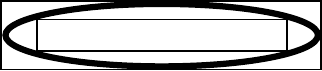}}\qquad
\subfloat[]{\includegraphics[width=0.25\textwidth]{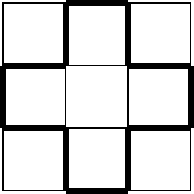}}\qquad
\caption{Values of $\sigma$ for some geometric objects \textbf{(a)} $\sigma=7/4$ for equilateral triangle \textbf{(b)} $\sigma=\sqrt{2}$ for ellipse \textbf{(c)} $\sigma=3$ for star with all boundary edges of equal length}
\label{fig:fat}
\end{figure}

Note that many common geometric objects have small values of $\sigma$ (see \Cref{fig:fat}). We shall assume that the objects of $\FF$ all lie inside the bounding box $[0, K]^d$ and have a length of at least one in each dimension. 
Unlike rectangles, these objects may not preserve their structural properties (fatness, spherical shape, or even convex shapes) under nonuniform scaling. Thus we cannot use our scale preparation steps directly. However, the following lemma asserts that since the objects are $\sigma$-rectangular, replacing them by their circumscribing rectangles does not incur much loss.
To state the lemma we need the following definitions. 
 Recall, an object $F$ is said to be $\sigma$-rectangular $(\sigma>1)$ if there exist axis-aligned inscribing and circumscribing hyperrectangles of $F$: $\In(F)$ and $\Out(F)$, respectively, such that for all $j\in[d]$, $l_j(\In(F))\geq  l_j(\Out(F))/\sigma$.     
Let $\FF$ be a set of $\sigma$-rectangular objects in $d$-dimension. For each object $F\in\FF$, let $\In(F)$ and $\Out(F)$ denote some inscribed and circumscribed axis-aligned hyperrectangles, respectively, with respect to which $F$ is $\sigma$-rectangular. Let $\In(\FF):= \{\In(F)\mid F \in \FF\}$ and $\Out(\FF):= \{\Out(F)\mid F \in \FF\}$.
%For any object $F\in \FF$, we shall let $\In(F)$ and $\Out(F)$ denote the inscribed and circumscribed axis-parallel hyperrectangles, respectively.

\begin{lemma}
\label{lem:fatobjects}
    Let $\FF$ be a set of $\sigma$-rectangular objects such that the side lengths of any two hyperrectangles in $\Out(\FF)$ lie within a factor of 2. Then $\opt(\FF)\le O(\sigma^d)\cdot \opt(\Out(\FF))$.
\end{lemma}

\begin{proof}
\begin{figure}[h]
    \centering
    \includegraphics[width=0.5\linewidth]{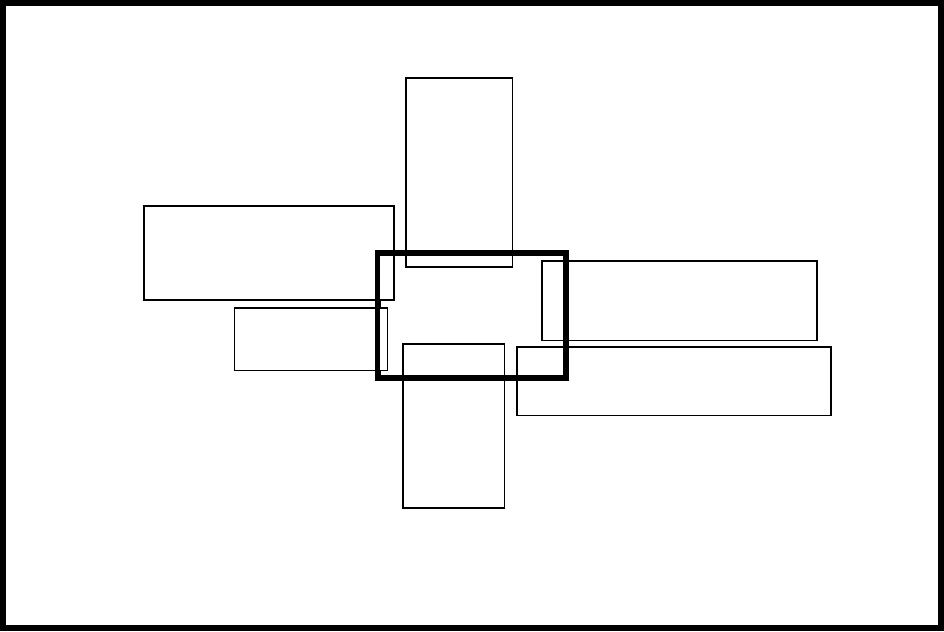}
    \caption{Figure for \Cref{lem:fatobjects}}
    \label{fig:boxBH}
\end{figure}

Consider the set OPT$(\FF)$. It suffices to show that each hyperrectangle in $\Out(\text{OPT}(\FF))$ overlaps with at most $O(\sigma^d)$ other hyperrectangles from the same set. Consider any hyperrectangle $H\in \Out(\text{OPT}(\FF))$ whose side lengths in the $d$ dimensions are $l_1,\ldots,l_d$, respectively. Let $F'\in \text{OPT}(\FF)$ be any other object such that $\Out(F')$ intersects $R$. Now since the side length of $\Out(F')$ is at most $2l_j$ in the $j^{\text{th}}$ dimension for all $j\in [d]$, the projection of $\Out(F')$ in dimension $j$ must completely lie within an interval of length at most $2\cdot 2l_j+l_j=5l_j$, centered at the midpoint of the $j^{\text{th}}$ side of $H$. Therefore $F'$ and hence $\In(F')$ must completely lie inside a box $B(H)$ whose side lengths in the $d$ dimensions are $5l_1,\ldots,5l_d$, respectively (see \Cref{fig:boxBH}). Also since the side length of $\Out(F')$ in dimension $j$ is at least $l_j/2$, for all $j\in [d]$, and $F'$ is $\sigma$-rectangular, $\In(F')$ has a side length of at least $l_j/(2\sigma)$ in dimension $j$. Since $\In(\text{OPT}(\FF)):= \{\In(F)\mid F\in \text{OPT}(\FF)\}$ is an independent set, by a volume argument, the number of $F'\in \text{OPT}(\FF)$ such that $\Out(F')$ overlaps $H$ is upper bounded by $(10\sigma)^d$. This completes the proof.
\end{proof}

The algorithm now is simple: whenever an object $F$ arrives, we replace it by its circumscribing hyperrectangle $\Out(F)$ and run our algorithm for hyperrectangles. Note that by our assumption, no hyperrectangle has a side of length less than one in any dimension, i.e., $D=0$ in this case. Now, our algorithm outputs an $O((\log K)^d\cdot \log\log K)$-competitive solution to the maximum independent set in $\Out(\FF)$, wherein the side length of the selected hyperrectangles lie within a factor of 2 in each dimension. By \Cref{lem:fatobjects}, this yields an $O((\sigma\log K)^d\cdot \log\log K)$-competitive solution to the maximum independent set of $\FF$. For $K=n^{(\log n)^{O(1)}}$, we thus obtain \Cref{thm:fatBounded}.
%%%%%End of Generalization%%%

%%%%%%%%%%%%%%RemovalK%%%%%%%%
\section{Removal of the Assumption on $K$}
\label{sec:KtofnMap}
In this section, we show how to get rid of the dependence on $K$ in the competitive ratio. We shall first consider the case of intervals and prove Theorem~\ref{thm:interUnbound} using Lemma~\ref{lem:bounded}. A key ingredient in our proof will be the following lemma which, roughly speaking, states that if we pick a random set $T\subseteq[n]$ of size $n/2$  uniformly at random, the $n/2 +1$ \textit{gaps} induced by $T$ on $[n]$ are each at most $4\log n$ with high probability.
\begin{restatable}{lemma}{gaps}
\label{lem:gaps}
Let $T$ be a subset of $[n]$ of cardinality $\lceil \frac n 2 \rceil$ chosen uniformly at random. Let 
$T=\{X_i \mid X_i \in [n], i\in \left[\lceil n/2 \rceil\right]\}$, where 
$1=X_0\leq X_1 < X_2 <\cdots < X_{\lceil \frac n 2 \rceil} \leq X_{\lceil\frac n 2\rceil + 1}=n$.  
\begin{align*}\text{Then,}\quad
\pr\left[ \max_{i\in \left[\lceil\frac n 2\rceil+1\right]} (X_{i}- X_{i-1}) \leq 4 \lceil\log_2 n\rceil \right] \geq 1 -\frac 1 {n}.
\end{align*}
\end{restatable}
See \Cref{sec:proofLemmaGaps} for the formal proof of this lemma. We briefly explain the intuition behind the proof. Roughly speaking, we divide $[n]$ into contiguous blocks of length $2 \log_2 n$ each. Observe that if $T$ hits all such blocks, the gaps induced by $T$ on $[n]$ are each of length at most $4 \log_2 n$. Now, since each element of $[n]$ is being included in $T$ with probability $1/2$, we can argue that the probability that $T$ does not intersect a fixed block is at most $1/ (2^{2 \log_2 n})=\frac 1{n^2}$. Applying a union bound over all such blocks, which are at most $n$ in number, we have that $T$ must hit all blocks with probability at least $1-\frac 1 n$.

We now prove Theorem~\ref{thm:interUnbound}.

\begin{proof}[Proof of Theorem~\ref{thm:interUnbound}]

Our idea is the following. We will divide the algorithm into two phases, each consisting of about $n/2$ intervals. In the first phase, we will not pick any of the intervals. Based on the intervals that arrived in the first phase, we will do a {\it non-uniform scaling} of the real line, so that the intervals arriving in the second phase when read in this scale will form a $(K,D)$-bounded instance for $K=O(n)$ and $D=O(\log n)$. In the second phase, we will run the algorithm in Lemma~\ref{lem:bounded} to get the desired competitive ratio.

Let $\mathcal I$ be the input instance with $|\mathcal I|$ = $n$.
Let $\mathcal I = \mathcal I_0\uplus\mathcal I_1$, where $I_0$ consists of the intervals received in the first phase, where $|I_0|=\lceil \frac n 2 \rceil$, and $I_1$ consists of the intervals received in the second phase.

\textbf{Non-uniform Scaling.} We first see how to obtain the desired non-uniform scaling.
For an interval $I$, let $\ell(I)$  denote its left endpoint. 
Let $I_1,\cdots, I_n$ be all the intervals in $\mathcal I$ such that $\ell(I_1) \leq \cdots \leq \ell(I_n)$.  Note that the $\ell(I_j)$'s need not be all distinct.  
Observe that the set of subscripts of the intervals in $\mathcal I_0$ is a random subset of $[n]$ of size $\lceil \frac n 2\rceil$ chosen uniformly at random. 
Let $\mathcal I_0 = \{I_{X_1},\cdots,I_{X_{\lceil n/2\rceil}}\}$ where $1=X_0\leq X_1 < \cdots <X_{\lceil n/2\rceil}\leq X_{\lceil n/2\rceil+1}=n$.
Let $\xi_1$ be the event $\max_{i\in\left[\lceil n/2\rceil + 1\right]} (X_i - X_{i-1}) \leq 4 \lceil\log_2 n\rceil$. Lemma~\ref{lem:gaps} yields that $\xi_1$ holds with probability $1-1/n$.
Since $I_j$'s are ordered by their left endpoints, if the event $\xi_1$ holds, then for each $i\in[\lceil n/2\rceil+1]$, there are at most $4\lceil\log_2 n\rceil$ intervals in $\mathcal I$ whose left endpoints are between $\ell(I_{X_{i-1}})$ and $\ell(I_{X_{i}})$. Let $\{p_1,\cdots, p_t\}=\{\ell(I_{X_1}),\cdots,\ell(I_{X_{\lceil n/2\rceil}})\}$ be the set of $t$ left endpoints of the intervals in $\mathcal I_0$, where $t\leq \lceil n /2 \rceil$ (we do not necessarily have an equality as some of the left endpoints might be the same) such that $p_1< \cdots < p_t$. In the second phase, we will not select any interval $I$ whose left endpoint is to the left of $p_1$ or right of $p_t$, i.e., $\ell(I) < p_1$ or $\ell(I) > p_t$; let $F\subseteq \mathcal I$ denote all such intervals that need to be ignored. From the above discussion, if $\xi_1$ holds, then $|F|\leq 8\lceil \log_2 n\rceil$.

\begin{figure} 
    \begin{subfigure}[t]{0.5\textwidth} 
        \centering 
        \includegraphics[width=0.9\textwidth]{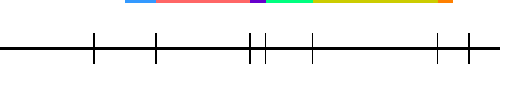} 
        \subcaption{Original interval} 
    \end{subfigure}\hfill 
    \begin{subfigure}[t]{0.5\textwidth} 
        \centering 
        \includegraphics[width=0.9\textwidth]{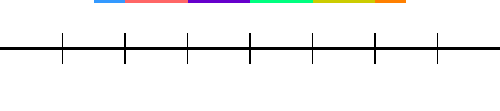} 
        \subcaption{After scaling} 
    \end{subfigure} 
    \caption{Non-uniform scaling. The vertical bars on the axis represent the $t$ starting points of the intervals in $\II_0$. After scaling, they become evenly spaced.}\label{fig:scaling} 
\end{figure}

We now define our new scale $s:[p_1,\infty) \rightarrow [1,t]$ as follows. 

\[
s(x) = 
\begin{cases}
i & \text{if }x = p_i \text{ for some } i\in[t]\\
i + \frac{x-p_i}{p_{i+1}-p_i} & \text{if } p_i < x < p_{i+1} \text{ for some $i\in[t-1]$}\\
t & \text{if } x > p_t
\end{cases}
\]

In the second phase, each interval in $\mathcal I_1\setminus F$, will be read in the scale $s$, i.e., the interval $[a,b]$ will be read as $[s(a),s(b)]$. See \Cref{fig:scaling} for an example. The function $s$ essentially scales the set $[0,K]$ such that the $t$ left endpoints of the intervals in $\mathcal I_0$ become equispaced. \Cref{fig:scaling}(b) illustrates how an interval in $\II_1\setminus F$ looks after the scaling operation. Additionally, intervals in $F$ will be all set to $[0,0]$. Observe that this scaling does not change the underlying intersection graph induced on $\mathcal I_1\setminus F$. Also, notice that  for all $[a,b]\in \mathcal I_1\setminus F$, $[s(a),s(b)]\subseteq[1,t]\subseteq[0,\lceil n/2\rceil]$. For each $i\in\{0\}\cup[t-1]$, let $D_i$ denote the number of intervals (read in the scale $s$ if they are not in $F$ and read as $[0,0]$ if they are in $F$) from $\mathcal I_1$ that are contained in $(i,i+1)$. Let $D= \max_{i\in\{0\}\cup[t-1]} D_i$. If $\xi_1$ holds, then $D \leq 4\lceil \log_2 n\rceil$ since an interval is contained in $(i,i+1)$ only if its left endpoint is contained in $(i,i+1)$. Thus, at the end of the first phase, if $\xi_1$ holds, we will end up with a $(K,D)$-bounded instance ($\mathcal I_1$ when scaled as explained), where $K=\lceil n/2\rceil$ and $D\leq 4\lceil \log_2 n\rceil$, on which we shall then run the algorithm in Lemma~\ref{lem:bounded}.

\textbf{Algorithm.} We now describe our algorithm. Let $\mathsf{Alg}$ denote the algorithm in Lemma~\ref{lem:bounded}.

\begin{itemize}
    \item{\textbf{First phase: Scale preparation.}}
    In this phase, the intervals in $\mathcal I_0$ arrive. Let $p_1 <\cdots < p_t$ be the distinct left endpoints of the intervals in $\mathcal I_1$. Define the scale $s$ as described above. 
    \item{\textbf{Second phase: Outsourcing.}} In this phase, the intervals from $\mathcal I_1$ arrive. We will feed these intervals to the algorithm in Lemma~\ref{lem:bounded}, namely $\mathsf{Alg}$, after appropriately scaling them with $s$. First, we feed $\mathsf{Alg}$ with the size of the instance it will receive: $n-\lceil n/2\rceil$. On receiving the interval $[a,b]\in\mathcal I_1$, 
    \begin{itemize}
        \item If $[a,b]$ is the last interval and no interval has been selected before, select $[a,b]$.
        \item Else, if $[a,b]\in F$, i.e., $a<p_1$ or $a>p_t$, we feed $\mathsf{Alg}$ the interval $[0,0]$, but do not select $[a,b]$ even if $\mathsf{Alg}$ selects $[0,0]$.
        \item Else, feed the interval $[s(a),s(b)]$ to $\mathsf{Alg}$ and select the interval $[a,b]$ iff $[s(a),s(b)]$ is selected by $\mathsf{Alg}$.
    \end{itemize}
\end{itemize}

\textbf{Analysis.}
Note that the intervals selected by the algorithm form an independent set. 
Let $h$ be the size of this independent set.  Observe that, if we condition on $\mathcal I_1$ being a particular $(K=\lceil n/2\rceil,D)$-bounded instance and $F$ being a particular set, then from Lemma~\ref{lem:bounded}, we conclude that for all large $K$ and $n$, the above algorithm will output an independent set of size  $h\geq {\lambda
}\cdot \opt(\mathcal I_1\setminus F)/{(D+\log K\log\log K)} - 1$ with probability at least $1 -  \mu/ {\log K}$, where $\lambda\in(0,1)$ and $\mu\geq 1$ are absolute constants. The $-1$ term appears in the above guarantee because $\mathsf{Alg}$ might select the interval $[0,0]$ fed to it corresponding to the intervals in $F$, but our algorithm ignores the intervals in $F$.

 We assume $\opt(\mathcal I) \geq \frac{10^3}\lambda \log n\cdot\log\log n$. Otherwise, since our algorithm will pick at least one interval, we already get a competitive ratio of $\frac{10^3}\lambda \log n\cdot\log\log n$.

 We have already established that the event $\xi_1$ implies $D\leq 4\lceil \log_2 n\rceil$ and $|F|\leq 8\lceil\log_2 n\rceil$, and $\pr(\xi_1)\geq1-\frac 1 n$. 
 To show our result, the only ingredient missing is a lower bound on $\opt(\II_1)$. To this end, let us define $\xi_2$ to be the event $\opt(\mathcal I_1) \geq \frac{1-\delta} 2\cdot \opt(\mathcal I)$, where $\delta =\frac 1{10}$. 
Now a similar calculation as in the proof of Lemma~\ref{lem:bounded} (using Lemma~\ref{lem:hoeffding}), yields for all $n\geq 100$, $\pr(\xi_2)\geq 1 - \frac 1{n^2}$.
Using union bound, the probability that both $\xi_1$ and $\xi_2$ hold is 
\[\pr(\xi_1 \wedge \xi_2)\geq 1 - \left(\frac 1 n + \frac 1{n^2}\right)\geq 1 - \frac 2 n.\]

$\xi_1\wedge\xi_2$ implies that $\mathcal I_1$ is a $(K=\lceil n/2\rceil, D = 4\lceil\log_2 n\rceil)$-bounded instance, and $|F| \leq 8\lceil \log_2 n\rceil$, $\opt(\mathcal I_1)\geq \frac{9}{20}\opt(\mathcal I)$.
Assuming these ranges for $K,D,|F|,$ and $ \opt(\mathcal I_1)$, we lower bound the expression ${\lambda}\cdot \opt(\mathcal I_1\setminus F)/{(D+\log K\log\log K)} - 1$.

 \begin{align*}
\frac{{\lambda}\cdot\opt(\mathcal I_1\setminus F)}{{(D+\log K\log\log K)}} - 1 
     &\geq \frac{{\lambda}\cdot \left(\frac 9{20}\cdot\opt(\mathcal I) - 8\lceil\log_2 n\rceil\right)}{{(4\lceil  \log_2 n\rceil+\log \lceil \frac n 2\rceil \cdot \log\log \lceil \frac n 2\rceil)}} - 1\\
&\geq\frac{ \lambda \cdot \frac 8{20}\cdot \opt(\mathcal I)}{10\log n\cdot \log\log n} - 1 
     \geq \frac{ 7\lambda \cdot\opt(\mathcal I)}{200\log n\cdot \log\log n},
 \end{align*}
where the second inequality holds since $\opt(\mathcal I)\geq 160\lceil \log_2 n\rceil$ and the last inequality holds since $ \lambda\cdot \opt(\mathcal I)\geq 200\log n\cdot \log \log n$.

Now, we show that the output of our algorithm $h$ is at least this lower bound with high probability. 

\begin{align*}
    \pr\left[h\geq \frac{ 7\lambda \cdot\opt(\mathcal I)}{200\log n\cdot \log\log n}\right] &\geq\pr(\xi_1\wedge\xi_2)\cdot \pr\left[h\geq \frac{ 7\lambda \cdot\opt(\mathcal I)}{200\log n\cdot \log\log n}\mid \xi_1\wedge\xi_2\right]\\
    &\geq \left(1 -\frac 2 n\right)\left(1 - \frac \mu{\log\lceil n/2\rceil}\right) \geq 1 - \frac{\mu+3}{\log n}.
\end{align*}\end{proof}
The algorithm can be implemented in $O(n\log n)$ time using BST (see \Cref{sec:datastructuresforInterval}).

\subsection{Extension to Hyperrectangles}
We now consider the case of hyperrectangles and prove \Cref{thm:hyperUnbound}.
As in the case of intervals, we shall have two phases, the scale preparation and outsourcing phases each comprising $n/2$ intervals each. After all hyperrectangles in the scale preparation phase have arrived, we perform non-uniform scaling in each dimension $j \in [d]$ independently. \Cref{lem:gaps} then implies that
the scaled hyperrectangles form a $(K,D)$-bounded set with $K=O(n)$ and $D=O(\log n)$, with probability at least $1-d/n$ (by taking union bound over all dimensions). \Cref{lem:generalBounded} then yields the promised competitive ratio guarantee in \Cref{thm:hyperUnbound}. The implementation details for achieving a running time of $\Tilde{O}_d(n)$ are deferred to \Cref{sec:datastructuresforRect}.

%%%%%%%%%End of RemovalK%%%%%

%%%%%%%%%%%%Conclustion%%%%%%
\section{Conclusion}\label{sec:conclusion}
We provide a guarantee of polylogarithmic strongly competitive ratio for the independent set in geometric intersection graphs for a wide range of objects. Note that we did not try to optimize the constants in the competitive ratio. With a refined analysis, the constants can be improved.  Apart from being the only known algorithm with a sublinear strongly competitive ratio, due to near-linear runtime, our algorithm also provides a simple candidate algorithm to be used in practice.

%%%%%%%%%End of Conclusion%%%

\section*{Acknowledgement}
The authors wish to thank Jaikumar Radhakrishnan for generously sharing his insights on hypergeometric concentration bounds, Rahul Saladi for an elegant proof of \cref{lem:DSy}, and K.V.N. Sreenivas for some helpful discussions.
\bibliographystyle{plain}
\bibliography{ref}

\appendix

%%%%%%%%%Appendix%%%%%%%%%%%%
\section{Proofs of Propositions}\label{sec:proofsOfProps}

\begin{proposition}\label{prop:sizeObliviousLowerBound}
    For all $\eps>0$, no online algorithm for interval scheduling in the random-order model which is not provided $n$, the size of the input instance, can be strongly  $O(n^{1-\eps})$-competitive.
\end{proposition}
\begin{proof}
We fix an online algorithm $\alg$ for interval scheduling which is strongly $c n^{1-\epsilon}$-competitive for some constant $c\in(0,1)$, and derive a contradiction. From the definition of strongly competitive, for all instances $\mathcal I$, 
$$\Pr\left[|\alg (\mathcal I)|\geq \frac {\opt(\mathcal I)}{c\cdot|\mathcal I|^{1-\epsilon}}\right] \geq 1-o(1).$$ 
Therefore, there exists $n_0$ such that for all instances $\mathcal I$ such that $|\mathcal I| \geq n_0$, 
$$\Pr\left[|\alg (\mathcal I)|\geq \frac {\opt(\mathcal I)}{c\cdot|\mathcal I|^{1-\epsilon}}\right] \geq \frac 1 2.$$ 
Fix such an $n_0$. Thus, for all instances $\mathcal I$, if $|\mathcal I|\geq n_0$ and  $\opt(\mathcal I) >0$, then $|\alg(\mathcal I)|>0$ with probability at least $\frac 1 2$. We will use this fact crucially to derive our contradiction. 

We create two instances $\mathcal I_1$ and $\mathcal I_2$.
Consider  some $2cn^{1-\epsilon}$ disjoint  intervals ($n >> n_0$), which we refer to as ${\it short}$ intervals, and a ${\it long}$ interval that intersects all the $2cn^{1-\eps}$ short intervals. Our instance $\mathcal I_1$ will consist of $n_0$ copies of the long interval, and $\mathcal I_2$ consists of all short intervals and $n-2cn^{1-\epsilon}$ copies of the long interval. Thus, $|\mathcal I_1|=n_0$ and $|\mathcal I_2|=n$, and $\opt(\mathcal I_1)=1$ and $\opt(\mathcal I_2)=2cn^{1-\epsilon}$. Furthermore, if $\alg$ on $\mathcal I_2$ picks a long interval, then, it cannot pick any other intervals resulting in $\alg(\mathcal I_2)=1$. We argue that $\alg$ on $\mathcal I_2$ will pick a long interval during the first $n_0$ intervals revealed to it with constant probability, giving a competitive ratio worse than $c n^{1-\epsilon}$ with constant probability. With constant probability $\mathcal I_1$ and $\mathcal I_2$ are indistinguishable for the first $n_0$ steps; this probability is 
$$ \prod_{i=0}^{n_0-1} \frac{n-i-2cn^{1-\epsilon}}{n-i}=(1-o(1))^{n_0}=1-o(1).$$
Now, from the observation above, $\alg$ on $\mathcal I_1$ picks an interval with probability at least $\frac 1 2$. Thus, $\alg$ on $\mathcal I_2$ picks an interval in the first $n_0$ steps with probability at least $\frac 1 2 (1-o(1)) \geq \frac{49}{100}$.
    
\end{proof}

\begin{proposition}\label{prop:observationLowerBound}
    No online algorithm for interval scheduling in the random-order model that does not pick any of the initial $c n$ intervals (for some absolute constant $c\in (0,1)$) can be strongly $O(1)$-competitive.
\end{proposition}
\begin{proof}

To see this, assume there is a strongly $d$-competitive algorithm for interval scheduling for some constant $d>0$. We construct an instance with optimum value $2d$ consisting of $2d$ disjoint intervals where the input consists of one copy of each of the first $2d-1$ intervals and $n-(2d-1)$ copies of the last interval. Since the algorithm is strongly $d$-competitive, with high probability, it must pick at least $2$ disjoint intervals. But, since it does not pick any of the first $cn$ intervals, with constant probability (roughly $c^{2d-1}$ for large $n$), all the intervals that have a single copy will appear in the initial phase, and will not get picked; hence, the algorithm can output at most one interval, a contradiction.
\end{proof}

\begin{figure}[h]
    \centering
    \includegraphics[width=0.3\linewidth]{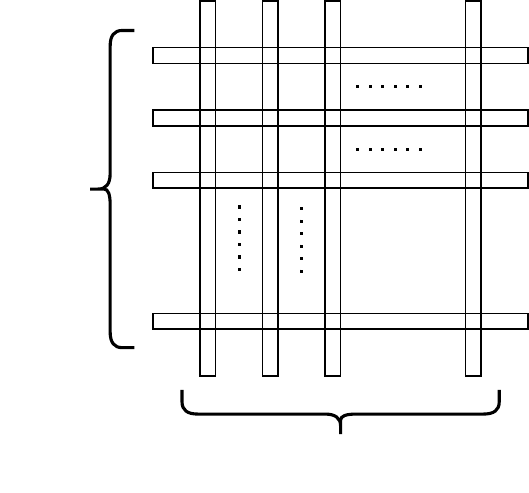}
    \caption{Instance with large inductive independence number}
    \label{fig:instance}
\end{figure}

\begin{proposition}
\label{thm:inducindep}
    The inductive independence number of axis-aligned rectangles is $\Omega(n)$.
\end{proposition}
\begin{proof}
    Consider the instance shown in Figure \ref{fig:instance}, consisting of $n/2$ ``tall and thin" rectangles (call them $V$) and $n/2$ ``short and wide" rectangles (call them $H$). Let $\prec$ be any ordering of the rectangles in $H\cup V$ in which the first rectangle (say $R$) comes from $H$. Then note that rectangles in $V$ form an independent set wherein each rectangle comes after $R$ in $\prec$. Since $R$ intersects every rectangle in $V$, the inductive independence number of the instance is at least $n/2$.    
\end{proof}
\section{Hypergeometric Concentration}\label{sec:hoeffding}

\hoeffding*
\begin{proof}
It is shown in~\cite{Chvatal79}, that for $t\geq 0$ and $0<(p+t)<1$,
\begin{align}\pr[X\geq(p+t)n] \leq \left( \left(\frac{p}{p+t}  \right)^{p+t}\left(\frac{1-p}{1-p-t}  \right)^{1-p-t}\right)^n.\label{eqn:KL}
\end{align}
The bounds claimed can now be derived from this bound and is a routine exercise.

We first prove the first inequality: $\pr[X \geq (1+\delta) pn]\leq \exp(-\delta^2 pn /3)$. We will prove this inequality by applying Bound~\eqref{eqn:KL}. This bound will not be applicable for two corner cases, which we first prove separately.

\begin{itemize}
\item{Case:} 
If $(1+\delta)p = 0$, then $p=0$, and the inequality trivially is true as its right hand side becomes $1$.
\item{Case:} Now, if $(1+\delta)p = 1$, then $p = \frac 1{1+\delta}$.
Now, if the left hand side probability is $0$, we are trivially done. Otherwise, $\Pr[X=(1+\delta)pn] = \Pr[X=n]=\prod_{i=0}^{n-1} \frac{M-i}{N-i}\leq (\frac M N)^n\leq p^n = (\frac 1{1+\delta})^n=(1-\frac \delta{1+\delta})^n \leq \exp(-\frac {\delta n}{1+\delta})\leq \exp(-\frac{\delta^2 n}{3}),$ where the last inequality follows from observing $\frac \delta{1+\delta} \geq \frac{\delta^2}{3}$ for all $\delta\in[0,1]$.

\item{Case:}
Now, for the remaining case, we set $t=\delta p$ and plug it in Bound~\eqref{eqn:KL}:
\begin{align*}\pr[X\geq(1+\delta)pn] &\leq  \left( \left(\frac{1}{1+\delta}  \right)^{(1+\delta)p}\left(\frac{1-p}{1-(1+\delta)p}  \right)^{1-(1+\delta)p}\right)^n\\
&= \left( \left(\frac{1}{1+\delta}  \right)^{(1+\delta)p}\left(1 +\frac{\delta p}{1-(1+\delta)p}  \right)^{1-(1+\delta)p}\right)^n\\
\end{align*}
Now, since $1+x\leq e^x$ for all $x\in \mathbb R$, we get
\begin{align*}\pr[X\geq(1+\delta)pn] &\leq  \left( \left(\frac{1}{1+\delta}  \right)^{(1+\delta)p}\left(e\right)^{\delta p}\right)^n
= \left(\frac {e^\delta}{(1+\delta)^{1+\delta
}}\right)^{pn}.
\end{align*}
Now, to finish the proof of the first part, we will show $\frac {e^\delta}{(1+\delta)^{1+\delta
}}\leq e^{-\delta^2/3}$, or equivalently by taking $\log$'s on both sides, $f(\delta):=\delta - (1+\delta)\log(1+\delta)+\frac{\delta^2}{3}\leq 0$. Taking derivative with respect to $\delta$, we get $f'(\delta)=1 -\log(1+\delta)-1+\frac{2\delta}{3}=-\log(1+\delta)+\frac{2\delta}{3}$. Taking derivative with respect to $\delta$ once more, we get $f''(\delta)= -\frac{1}{1+\delta} +\frac{2}{3}$. 

Now, for $\delta\in[0,\frac 1 2)$, $f''(\delta)<0$ and  for $\delta >\frac 1 2$, $f''(\delta)> 0$. Thus, $f'(\delta)$ decreases in the interval $[0,\frac 1 2)$ and then increases in the interval $(\frac 1 2, 1]$. Since, $f'(0)=0$ and $f'(1)= -\log 2 +\frac 2 3 <0$, $f'(\delta)\leq 0$ for $\delta\in[0,1]$. Therefore, $f(\delta)$ never increases in the interval $[0,1]$. Now since $f(0)=0$,  we have $f(\delta)\leq 0$ for $\delta\in[0,1]$, as desired. Thus,
$$\pr[X\geq(1+\delta)pn] \leq \exp(-\delta^2pn/3).$$
\end{itemize}
We now prove the second inequality: $\pr[X \leq (1-\delta) pn]\leq \exp(-\delta^2 pn /2).$ Like in the proof of the first inequality above, we first handle two corner cases that do not follow from Bound~\eqref{eqn:KL}.

\begin{itemize}
    \item{Case:} 
If $(1-\delta)p = 0$, then either $p=0$ or $\delta = 1$. If $p=0$, then again the inequality follows trivially as the right hand side becomes $1$. Thus, we assume $\delta = 1$. Now, if the left hand side probability is $0$, we are trivially done. Otherwise, the probability $\Pr[X\leq (1-\delta)pn] = \Pr [X = 0] = \prod_{i=0}^{n-1} \frac{M-N-i}{M-i} \leq (\frac {M-N}{M})^n \leq (1-p)^n \leq \exp(-pn) \leq \exp(-\frac{\delta^2 p n}{2})$, where the last inequality follows from observing $1\geq \frac{\delta^2}2$ for $\delta = 1$.
\item{Case:} Next, if $(1-\delta)p = 1$ then $p=1$ and $\delta = 0$. Thus, again the inequality follows trivially, as the right hand side becomes $1$.
\item{Case:} To prove the remaining case, let us say that the non-red balls are blue, and the fraction of blue balls is $q=1-p$. Let $Y$ be the number of blue balls that appear in the sample. Then, for $t\geq 0$,
\begin{align*}
\pr[Y\geq (q+t)n])&\leq \left( \left(\frac{q}{q+t}  \right)^{q+t}\left(\frac{1-q}{1-q-t}  \right)^{1-q-t}\right)^n\\
&=\left( \left(\frac{1-p}{1-p+t}  \right)^{1-p+t}\left(\frac{p}{p-t}  \right)^{p-t}\right)^n.
\end{align*}
Setting $t=\delta p$, we have,
$\pr[X\leq (1-\delta)pn] \leq \pr[Y\geq n-(1-\delta)pn] = \pr[Y\geq (q+t)n]$.
Thus, 
\begin{align*}
    \pr[X\leq (1-\delta)pn] &\leq \left( \left(\frac{1-p}{1-p+t}  \right)^{1-p+t}\left(\frac{p}{p-t}  \right)^{p-t}\right)^n\\
    &= \left( \left(1 -\frac{t}{1-p+t}  \right)^{1-p+t}\left(\frac{1}{1-\delta}  \right)^{(1-\delta)p}\right)^n\\
    &\leq \left( \frac{e ^{-\delta}}{\left({1-\delta}  \right)^{(1-\delta)}}\right)^{pn},
\end{align*}
where the last inequality follows by observing $1+x \leq e^x$ for all $x\in \mathbb R$.
Now, to finish the proof, we will show $\frac{e ^{-\delta}}{\left({1-\delta}  \right)^{(1-\delta)}}\leq e^{-\delta^2/2}$, or equivalently by taking $\log$'s on both sides, we will show $g(\delta):=-\delta -(1-\delta)\log (1-\delta) +\frac{\delta^2}{2}\leq 0$. Taking derivative with respect to $\delta$, we have $g'(\delta)= -1 +\log(1-\delta)+1 +\delta=\log(1-\delta) +\delta$. Taking derivative once more, $g''(\delta) = -\frac 1{1-\delta} + 1$.

Now, $g''(\delta) \leq 0$ for $\delta \in [0,1)$. Thus, $g'(\delta)$ is non-increasing in the interval $[0,1)$. Since, $g'(0)=0$, we have $g'(\delta) \leq 0$ for $\delta\in [0,1)$. Thus, $g(\delta)$ is non-increasing in the interval $[0,1)$. Now, since $g(0)=0$, we have, $g(\delta)\leq 0$ for $\delta\in[0,1)$. 
Thus, for all $\delta\in[0,1)$,
$$\pr[X\leq (1-\delta)pn] \leq \exp(-\delta^2 pn/2).$$

For the special case of $\delta = 1$, we need to show $\pr[X=0]\leq \exp(-pn/2)=\exp(-\frac{Mn}{2N})$. Now, if $n > N-M$, $\pr[X=0]=0$, and we are trivially done. Otherwise,
\begin{align*}
    \pr[X=0] &= \prod_{i=0}^{n-1} \frac{N-M-i}{N-i}=\prod_{i=0}^{n-1}\left(1 - \frac {M}{N-i}\right)\leq \exp\left(-M\sum_{i=0}^{n-1} \frac 1{N-i}\right)\\
    &=\exp\left(-\frac M N\sum_{i=0}^{n-1} \frac N{N-i}\right)\leq \exp\left(-\frac {Mn} N\right)\leq \exp\left(-\frac {Mn} {2N}\right).
\end{align*}
\end{itemize}
\end{proof}
\section{Extension to Bounded Hyperrectangles}\label{sec:boundedExtensionProof}
We now extend our result for online interval scheduling to the maximum independent set of $(K,D)$-bounded hyperrectangles. We first need the definition of $(K,D)$-bounded hyperrectangles.

\begin{definition}\label{def:boundedHyperrectangles}
   We say a set of $d$-dimensional axis-aligned hyperrectangles $\mathcal H$ is $(K,D)$-bounded if each $H\in \mathcal H$ is a subset of $[0,K]^d$ and for all $j\in[d]$ and $i \in\mathbb{Z}$,  the number of hyperrectangles whose projection along dimension $j$ is completely  contained in $(i, i + 1)$ is at most $D$, and $(\log K)^d \geq D$.
\end{definition}

Note that the condition $(\log K)^d \ge D$ is purely technical. We show that there exists an online algorithm which on $(K,D)$-bounded $d$-dimensional instances is $O_d(D^2 +  (\log K)^d\cdot \log\log K)$-competitive with probability $1-o_K(1)$. More precisely, we have the following lemma.
\generalBounded*
\begin{proof}[Proof of~\cref{lem:generalBounded}]

We first set up some definitions and notations that will help in describing and analyzing our algorithm.

Let $\mathcal H$ be a set of $d$-dimensional hyperrectangles that forms a $(K,D)$-bounded instance that is given as input and $|\mathcal H|=n$. Let $\epsilon>0 $ be a fixed constant. Our analysis works for a range of values for $\epsilon$; for concreteness, we set $\epsilon = 1$. Also, let
$k:= \lceil \log_{1+\epsilon} K\rceil.$

We partition $\mathcal H$ into $d+k^d  \leq (k+1)^d$ sets: $S_x$'s and $S_y$'s, where each $x$ is an index in $[d]$ and each $y$ is a $d$-tuple in $[k]^d$, depending on the side lengths of the hyperrectangles in $\mathcal H$. To define this partition, we partition the range of side lengths of the hyperrectangles in $\mathcal H$, i.e., $[0,(1+\epsilon)^k] \supseteq [0,K]$, into $k+1$ intervals: $\range_0,\cdots, \range_k$ as follows.
\[\range_0 := [0,1], \
\range_i := ((1+\epsilon)^{i-1},(1+\epsilon)^i]\text{ for all $i\in [k]$.}\]
Note that 
$[0,K]\subseteq \range_0 \uplus \cdots \uplus \range_k.$
Now, $S_x$'s and $S_y$'s can be defined as follows.
\[S_x := \{ H \in \mathcal H \mid x \text{ is the smallest index in } [d] \text{ such that } l_x(H) \in \range_0\} \text{ for all $x\in[d]$.}\]
\[S_y := \{ H\in \mathcal H \mid \text{ for all } j\in [d] \ l_y(H) \in \range _{x_y}\} \text{ for all $y \in [k]^d $.}\]
Note that the hyperrectangles in $S_x$'s have at least one side of length at most one, whereas the hyperrectangles in $S_y$'s have all their side lengths strictly more than one. Observe that
$\mathcal H$ is the disjoint union of the $S_x$'s and the $S_y$'s. 

For a set of hyperrectangles $S$ (typically $S\subseteq \mathcal H$), we use $\OPT(S)$ to denote a maximum independent set of hyperrectangles in $S$, and $\opt(S):=|\OPT(S)|$. Finally, for ease of notation, we use $x$ to denote an arbitrary subscript in $[d]$ and $y$ to denote an arbitrary subscript in $[k]^{d}$. We use the subscript $i$ for both $x$'s and $y$'s. The sets these subscripts belong to will not be stated explicitly and should be understood based on the letter used: $x$, $y$, or $i$. 

We now describe our algorithm before moving to its analysis.

\noindent\textbf{Algorithm.} 
     The algorithm has two phases: the \textit{observation} phase, when it receives the first $\lceil \frac n 2 \rceil$ hyperrectangles, and the \textit{action} phase, when it receives the remaining $n-\lceil\frac n 2 \rceil$ hyperrectangles. 
    Let $\mathcal H_0$ denote the set of hyperrectangles received during the observation phase and $\mathcal H_1$  the set of hyperrectangles received during the action phase. 
    Note that $\mathcal H=\mathcal H_0 \uplus \mathcal H_1$, where $\mathcal H_0$ is a subset of $\mathcal H$ of size $\lceil\frac n 2\rceil$ chosen uniformly at random, and $\mathcal H_1$ is a subset of $\mathcal H$ of size $n-\lceil\frac n 2\rceil$ chosen uniformly at random.\footnote{We often do not say with probability one or almost surely when it is clear from the context.} 
    
    \begin{itemize}
    
    \item{\textbf{Observation phase}} In this phase hyperrectangles in $\mathcal H_0$ arrive.
    \begin{itemize}
    \item{Initialization:} Initialize $G_i=\emptyset$ for all $i$. 
    \item When a new hyperrectangle $H\in \mathcal H_0$ arrives: Add $H$ to $G_i$ iff $G\in S_i$ and $H$ does not intersect any of the hyperrectangles in $G_i$.
    \item At the end of the observation phase, $\hat L_i:= |G_i|$.
    \item
   Compute the index $m$ as follows. Let $m_1$ be an index $x$ such that $\hat L_{m_1} \geq \hat L_x$ for all $x$. Let $m_2$ be an index $y$ such that $\hat L_{m_2} \geq \hat L_y$ for all $y$. If $\hat L_{m_1} > \frac{(k+1)^d}D \cdot \hat L_{m_2}$, set $m=m_1$, else set $m=m_2$.
    \end{itemize}
    \item{\textbf{Action phase}} In this phase, the hyperrectangles in $\mathcal H_1$ arrive. 
    \begin{itemize}
        \item{Initialization:} $R_m=\emptyset$. On arrival of a hyperrectangle $H\in \mathcal H_1$, do the following.
    \item If $H\in S_m$ and $H$ does not intersect any hyperrectangles in $R_m$, select $H$ and add $H$ to $R_m$.
    \item Else, if
    $H$ is the last interval to arrive and $R_m=\emptyset$, then select $H$.
    \item Else, if $H\notin S_m$, then do not select $H$.
    \end{itemize}
    \end{itemize}
 In the observation phase, the algorithm uses the greedy algorithm to estimate $\opt(L_i)$, where $L_i = S_i\cap \mathcal H_0$, for all $i$. Thus, $\hat L_i$ is an estimate of $\opt(L_i)$.
 Then, 
 it either sets $m$ to be an $x$ when $L_m\geq L_x$ for all $x$ and $\hat L_m >\frac{(k+1)^d}D\cdot\hat L_y$ for all $y$, else it  sets $m$ to be a $y$ such that $\hat L_m\geq \hat L_x$ for all $x$. In the action phase, the algorithm simply runs the greedy algorithm $\greedy$ on $R_m := \mathcal H_1 \cap S_m$, ignoring all hyperrectangles not in $S_m$. When the algorithm receives the very last hyperrectangle and $R_m = \emptyset$, it picks the last hyperrectangle. This ensures that the algorithm outputs an independent set of size at least one whenever $n
    \geq 2$.
    
We now analyze the competitive ratio of the above algorithm.

\noindent\textbf{Analysis.}
We assume that $\opt(\mathcal H) \geq 10^4\cdot d\cdot 4^{d+1}\cdot (k+1)^d\cdot \log k$, else since the algorithm is guaranteed to pick at least one hyperrectangle (for $n\geq 2$), we get a  competitive ratio of $10^4\cdot d\cdot 4^{d+1}\cdot (k+1)^d\cdot \log k$.

We now define an event $E$ and show that (I) if $E$ holds, then the algorithm achieves a competitive ratio of $O(D^2+(k+1)^d)$, and (II) $E$ holds with probability $1-o_k(1)$.
In order to define the event $E$, we first define $L_i = \mathcal H_0 \cap S_i$, $\hat L_i = |\greedy(L_i)|$, and $R_i = \mathcal H_1 \cap R_i$ for all $i$, as explained above. Furthermore, let us define $B=\{i \mid \opt(S_i) \geq 10^3 d\log k\}$.

We say that the event $E$ holds iff for all $i\in B$, \begin{align}\min(\opt(L_i),\opt(R_i))\geq (1-\delta)\frac {\opt(S_i)} 2,\quad\text{where $\delta = \frac 1{10}.$}\label{eqa} \end{align}

We first show that (I) is true.

\noindent\textbf{(I) Assuming $E$ holds, the algorithm is $O(D^2+(k+1)^d)$-competitive.}

We first prove four basic bounds that will help us in the proof.

First, we want to show that $\sum_{i\in B} \opt(S_i)$ is large, i.e., at least a constant fraction of $\opt(\mathcal H)$. Intuitively, this implies that even if an algorithm ignored hyperrectangles in $\cup_{i\notin B} S_i$, it would have substantial value available.
\begin{align}
     &\opt(\mathcal H)\leq \sum_{i\in B} \opt(S_i) + \sum_{i\notin B} \opt(S_i) \nonumber\\
    \implies &\opt(\mathcal H) -\sum_{i\notin B} \opt(S_i)\leq \sum_{i\in B} \opt(S_i)\nonumber\\
    \implies &\frac{\opt(\mathcal H)}2 \leq\sum_{i\in B} \opt(S_i),\label{eqb}
\end{align}
where the last inequality holds since $\sum_{i\notin B} \opt(S_i) \leq (k+1)^d\cdot10^3\cdot d\log k\leq  \frac{\opt(\mathcal H)}2$.

Next, we want to show for $i\in B$, $\opt(R_i)$ is at least a constant fraction of $\opt(L_i)$.
For each $i\in B$, 
\begin{align}\opt(R_i) \geq \frac{1-\delta}2\opt(S_i)\geq \frac{1-\delta}2 \opt(L_i),\label{eqc}\end{align}
where the last inequality holds as $L_i \subseteq S_i$.

We also have the following guarantees on the $\greedy$ algorithm for $(K,D)$-bounded $d$-dimenional hyperrectangles.

\begin{align}\opt(L_x)\geq\hat L_x \geq \frac{\opt(L_x)}{3D}\label{eqe}\end{align}

\begin{align}\opt(L_y)\geq\hat L_y \geq \frac{\opt(L_y)}{4^d}\label{eqd}\end{align}

We now prove the following claim.
\begin{claim}
    $m\in B$, where $m$ is the index picked by the algorithm at the end of the observation phase.
\end{claim}
\begin{proof}
  We assume $m\notin B$ and derive a contradiction. Combining inequalities \eqref{eqa} and \eqref{eqb} we have,

\begin{align*}\sum_{i\in B}\opt(L_i)\ge \frac{1-\delta}{4}\opt(\mathcal H)&\ge \frac{(1-\delta)\cdot 10^4}{4}  \cdot d\cdot 4^{d+1}\cdot (k+1)^d\cdot \log k\\
&=\frac{9}{4} \cdot 10^3\cdot d\cdot 4^{d+1}\cdot (k+1)^d\cdot \log k.\end{align*}
Whereas, if $m\not\in B$, then 
\begin{align*}
    &\sum_{i\in B}\opt(L_i)\leq\sum_{i} \opt(L_i)\leq
    3D\sum_x\hat L_x + 4^d\sum_x\hat L_y 
    \\
    &\leq d\cdot (3D)\cdot \frac{(k+1)^d}{D} \cdot 10^3\cdot d\cdot \log k
    +
    (k+1)^d\cdot 4^d \cdot 10^3\cdot d\cdot \log k\\
    &\leq 4^{d+1}\cdot 10^3\cdot d\cdot(k+1)^d\cdot\log k 
\end{align*}

since, $m\notin B$ and $(k+1)^d \geq D$ implies that the expression obtained in terms of $\hat L_i$'s is maximized when we set $m$ to be a $y$ and set largest possible values to the $\hat L_i$'s respecting $\hat L_x \leq \frac{(k+1)^d} D\cdot\hat L_m$ and $\hat L_y \leq \hat L_m \leq \opt(S_m)\leq 10^3\cdot d\cdot \log k$. Thus, combining the above two inequalities, we get $1\geq \frac 9 4$, a contradiction. Thus, $m\in B$.
\end{proof}

We now have two cases depending on whether the algorithm picks $m$ to be an $x$ or a $y$ at the end of the observation phase.
In either case, we will show that the output of the algorithm $|\greedy(R_m)|\geq \frac 1{O(D^2+(k+1)^d)}\opt(\mathcal H)$.

\noindent\textbf{Case: $m=x$ for some $x$.}
\begin{align*}
    \frac{\opt(\mathcal H)}2&\leq \sum_{i\in B} \opt(S_i)\quad\text{(from Bound~\eqref{eqb})}\\
    &\leq  \frac 2{1-\delta} \sum_{i\in B}\opt(L_i)\quad\text{(from Bound~\eqref{eqa})}\\
    &\leq \frac 2{1-\delta}\left(\sum_x\opt(L_x) + \sum_{y}\opt(L_y)\right)\\
    &\leq \frac 2{1-\delta}\left(3D\sum_x L_x + 4^d\sum_{y}\hat L_y\right)\\
    &\leq \frac 2{1-\delta}\left(d\cdot3D\cdot\hat L_m + (k+1)^d\cdot \frac {D\cdot 4^d}{(k+1)^d} \cdot \hat L_0\right),\\
\end{align*}
where the last inequality follows from the fact that the algorithm picks $m$ to be an $x$ when $\hat L_m \geq \hat L_x$ and $\hat L_m \geq \frac{(k+1)^d}{D}\hat L_y$ for all $x$ and $y$. Thus,
\begin{align*}
    \opt(L_m)\geq\hat L_m &\geq \frac{1-\delta}{4D(3d+4^d)}\opt(\mathcal H)\\
    \implies \opt(R_0) &\geq \frac{(1-\delta)^2}{8D(3d+4^d)}\opt(\mathcal H)\quad\text{(from Bound~\eqref{eqc}).}
\end{align*}

Now invoking Lemma~\ref{lem:greedyCompetitiveRatioprelim}, the size of the output is 
\[|\greedy(R_0)|\geq \frac 1{3D} \opt(R_0)\geq \frac{(1-\delta)^2}{24D^2(3d+4^d)}\opt(\mathcal H)\geq \frac{1}{100D^2 4^d}\opt(\mathcal H).\]

\noindent\textbf{Case: $m=y$ for some $y$.}
\begin{align*}
    \frac{\opt(\mathcal H)}2&\leq \sum_{i\in B} \opt(S_i)\quad\text{(from Bound~\eqref{eqb})}\\
    &\leq  \frac 2{1-\delta} \sum_{i\in B}\opt(L_i)\quad\text{(from Bound~\eqref{eqa})}\\
    &\leq \frac 2{1-\delta}\left(\sum_x\opt(L_x) + \sum_{y}\opt(L_y)\right)\\
    &\leq\frac 2{1-\delta}\left(3D\sum_x\hat L_x + 4^d\sum_{y} \hat L_y\right)\\
    &\leq \frac 2{1-\delta}\left( d\cdot3\cdot(k+1)^d \cdot\hat L_m + 4^d\cdot(k+1)^d \cdot \hat L_m \right),\\
\end{align*}
where the last inequality follows from the fact that $\hat L_x \leq \frac{(k+1)^d}D\cdot \hat L_m$  and $L_y \leq L_m$ for all $x$ and $y$ as $m$ is a $y$.

Thus,
\begin{align*}
    \opt(L_m) \geq \hat L_m &\geq \frac{1-\delta}{4 (3d+4^d)(k+1)^d} \opt(\mathcal H)\\
    \implies \opt(R_m) &\geq \frac{(1-\delta)^2}{8(3d+4^d)(k+1)^d}\opt(\mathcal H)\quad\text{(from~ Bound \eqref{eqc}).}
\end{align*}

Now invoking Lemma~\ref{lem:greedyBounded}, the size of the output is 
\[|\greedy(R_m)|\geq \frac 1{4^d}\cdot \opt(R_m)\geq \frac{(1-\delta)^2}{8\cdot 4^d(3d+4^d)(k+1)^d}\opt(\mathcal H)\geq \frac{1}{100 \cdot4^{2d}\cdot(k+1)^d}.\]

Thus, in either case, 
\[|\greedy(R_m)|\geq \frac{(1-\delta)^2}{100\cdot 4^{2d}(D^2+(k+1)^d)}\opt(\mathcal H).\]

Thus, our algorithm is $\max(10^4\cdot d\cdot 4^{d+1}\cdot (k+1)^d\cdot \log k, 100\cdot 4^{2d}(D^2+(k+1)^d))$-competitive, which implies that it is $O(d\cdot 5^{2d}(D^2+k^d\log k))$-competitive.

We now show that (II) is true.

\noindent\textbf{(II) $E$ holds with probability at least $1-(3/k)^d$.}

Fix an $i\in B$. 

In the following, we use the fact that $\mathcal H_0$ is a uniformly random subset of $\mathcal H$ of size $\lceil\frac n 2\rceil$, and upper bound the probability that only a small fraction of $\OPT(S_i)$ is picked in $\mathcal H_0$.

\begin{align*}
\pr\left[\opt(L_i) < (1-\delta)\frac{\opt(S_i)} 2\right] &\leq  \pr\left[\opt(L_i) \leq (1-\delta)\frac{\opt(S_i)} n \left\lceil \frac n 2 \right\rceil\right] \\
&\leq \exp\left(- \frac {\delta^2}2 \frac{\opt(S_i)}n \left\lceil \frac n 2\right\rceil\right) \\
&\leq \exp\left(- \frac {\delta^2}4 \opt(S_i)\right)\\
&\leq \exp\left(- \frac {10^3\delta^2 d}4 \log k\right)\\
&\leq \frac 1 {k^{2d}},
\end{align*}
where the second inequality follows from \Cref{lem:hoeffding} where $N, M, P,$ and $n$ in the Lemma statement are set to $n, \opt(S_i), \frac{\opt(S_i)}n,$ and $\lceil \frac n 2 \rceil$ respectively, and the fourth inequality follows since $\opt(S_i)\geq 10^3 d\log k$, as $i\in B$. 

A similar calculation, where we use the fact that $\mathcal H_1$ is a uniformly random subset of $\mathcal H$ of size $n-\lceil \frac n 2 \rceil$, yields for all $n \geq 100$,
\begin{align*}
\pr\left[\opt(R_i) < (1-\delta)\frac{\opt(S_i)} 2\right] 
&\leq \frac 1 {k^{2d}}.
\end{align*}

Now, from the union bound,
$\pr[\exists i\in B \text{ such that }\min(\opt(L_i),\opt(R_i)) < (1-\delta) \frac{\opt(S_i)}2] \leq \frac{2|B|}{k^{2d}} \leq \frac {2(k+1)^d}{k^{2d}}\leq (\frac 3 k)^d$, for all $k\geq 2$. Thus, $\pr[E] \geq 1 - (\frac 3 k)^d$.

(I) and (II) together complete the proof.
\end{proof}

\section{Proof of Lemma~\ref{lem:gaps}}\label{sec:proofLemmaGaps}
\gaps*
\begin{proof}[Proof of~\cref{lem:gaps}]
    We first define disjoint sets $R_1,\cdots, R_m$,
which we refer to as blocks. Each block is of size $b:=2\lceil \log_2 n \rceil$ and together they almost cover $[n]$.
    $R_1 = \{1, 2, \cdots,b\}$,
    $R_2 = \{b + 1, 
    b + 2, \cdots, 2b\}$, and so on, and finally
    $R_m = \{(m-1)b + 1, (m-1)b + 2, \cdots, mb\}$,  where $m$ is the largest integer such that
    $mb\leq n$. Thus, $[mb]=R_1 \uplus \cdots \uplus R_m$, and the number of elements of $[n]$ that are not covered by any of the blocks is $|[n]\setminus [mb]| < b$.

Let $T$ be a random subset of $[n]$ of size $\lceil \frac n 2\rceil$ chosen uniformly at random. Let $E$ be the event that $T$ {\it hits} all the blocks, i.e., for all $i \in [m]$, $T \cap R_i \neq \emptyset$.

Let the $X_i$'s be defined as in the lemma statement. Observe that if $E$ holds, then the {\it gaps} $X_i - X_{i-1}$
for each $i$ is at most $2b = 4 \lceil \log_2 n\rceil$.
In other words, if the event $E$ holds, then $\max_{i\in\left[\lceil n/2\rceil + 1\right]} (X_i - X_{i-1}) \leq 4 \lceil\log_2 n\rceil$.
Thus, to obtain the claimed bound, we will show that $E$ holds with high probability. 

To this end, we first show that the probability that a given block is not hit is low. In fact, we show a slightly stronger statement that the probability an arbitrary set of size $b$ is not hit is low.
Fix a set $R\subseteq[n]$ such that $|R|=b$. We want to upper bound the probability of the event $R\cap T = \emptyset$. Let $A$ be a random subset of $[n]$ of size $b$ chosen uniformly at random and independent of $T$. Notice that
$\pr[R \cap T=\emptyset]= \pr[A \cap T = \emptyset]$ as $\pr[A \cap T = \emptyset]$
\begin{align*}
    &= \sum_{S\subseteq[n] : |S|=b} \pr[A = S] \pr[S\cap T = \emptyset\mid A=S]
    =\sum_{S\subseteq[n]:|S|=b} \pr[A = S] \pr[R\cap T = \emptyset]\\
    &=\pr[R\cap T = \emptyset]\sum_{ S\subseteq[n]: |S|=b} \pr[A = S]
   =\pr[R\cap T=\emptyset],
\end{align*}
where $\pr[S\cap T = \emptyset \mid A = S] = \pr[R\cap T = \emptyset]$ follows from symmetry.

Now, we will evaluate $\pr[A\cap T=\emptyset]$. We will first condition on $T$ being a fixed set $T'$ of size $\lceil n/2\rceil$ and then evaluate the probability that $A$ does not hit $T'$. For this, we will think of $A$ as being generated by picking one element at a time uniformly at random from the unpicked elements. The probability that the first element picked in $A$ does not interest $T$ is $\frac{n-\lceil n/2 \rceil}n \leq \frac 1 2$. Next, the probability the second element picked in $A$ does not belong to $T$ conditioned on the first picked element not belonging to $T$ is again $\frac{n-\lceil n/2 \rceil - 1}{n-1} \leq \frac 1 2$, and so on. Thus,

\begin{align*}
    \pr[R\cap T=\emptyset] &= \pr[A\cap T = \emptyset]\\
    &= \sum_{T'\subseteq [n] : |T'|=\lceil \frac n 2\rceil} \pr[T=T']Pr[A\cap T' = \emptyset \mid T=T']\\
     &\leq \sum_{T'\subseteq [n] : |T'|=\lceil \frac n 2\rceil} \pr[T=T'] \left(\frac 1 2\right)^{2\lceil\log_2 n\rceil}\\
     &=\left(\frac 1 2\right)^{2\lceil\log_2 n\rceil}\leq \frac 1{n^2}.
\end{align*}

We can now obtain our result by simply applying the union bound as follows.
\begin{align*}
    \pr&\left[\max_{i\in\left[\lceil n/2\rceil + 1\right]} (X_i - X_{i-1}) \leq \lceil4 \log_2 n\rceil\right] \geq \pr[E]\\
    &\quad\quad\quad= 1 - \pr[\exists i\in[m] \text{ such that }R_i\cap T = \emptyset]\\
    &\quad\quad\quad\geq 1 - \sum_{i\in[m]}\pr[R_i\cap T = \emptyset]\quad\text{(from the union bound)}\\
    &\quad\quad\quad\geq 1 - \frac{m}{n^2}\geq 1- \frac 1 n,
\end{align*}
where the last inequality holds since the number of blocks $m \leq n.$
\end{proof}

%%%%%%%%%%Data structures%%%%
\section{Data Structures}
\label{sec:datastructures}
In this section, we show how to implement our algorithms in $\Tilde{O}(n)$ time. We begin with the easier case of intervals before generalizing to hyperrectangles.

\subsection{Intervals}
\label{sec:datastructuresforInterval}
The length of an interval used for classification during the observation phase is determined based on the relative position of the interval with respect to the intervals arriving during the scale preparation phase. Hence if we store the starting points of the intervals sorted in non-decreasing order in a balanced binary search tree (which can support insertion in $O(\log n)$ time), the scaled length of an interval in the outsourcing phase can be determined in $O(\log n)$ time using two \textsf{Search} operations on the binary search tree.

During the observation phase, in each size class, we store the arriving intervals sorted in the non-decreasing order of their ending points, in a binary search tree. At the end of this phase, the maximum independent sets of all the size classes can be determined in $O(n)$ time by applying $\greedy$ on each class independently ($\greedy$ runs in linear time when the intervals are sorted according to their ending points \cite{kleinberg2006algorithm}). Hence the class $m$ to proceed during the action phase can be found in $O(n)$ time at the end of the observation phase. Finally, during the action phase, we again implement $\greedy$ on the class $m$ using a binary search tree as described before. Therefore, the total running time of our algorithm is $O(n\log n)$, finishing the proof of \Cref{thm:interUnbound}.

\subsection{Hyperrectangles}
\label{sec:datastructuresforRect}
We now turn to the implementation for $d$-dimensional hyperrectangles. First note that checking whether two $d$-dimensional hyperrectangles intersect takes $O(1)$ time for constant $d$.

The scale preparation phase is a direct generalization of the case for intervals to higher dimensions. We maintain the starting points of the hyperrectangles in non-decreasing order in each dimension with the help of $d$ separate binary search trees, one for each dimension. When a hyperrectangle $H$ arrives, we insert the $d$ starting points of $H$ (one in each dimension) into the corresponding binary search trees. This takes $d\cdot O(\log n) = O_d(\log n)$ time per step.

Now during the observation phase, we essentially implement $\greedy$ for each size class independently. Hence we need a separate data structure for each class that maintains the independent set of hyperrectangles of that class that are selected by $\greedy$. When a new hyperrectangle $H$ in that class arrives, the data structure must determine whether $H$ intersects any previously selected hyperrectangle. If not, then $H$ needs to be inserted into the data structure. We shall call this as an \textit{independence update} for $H$.

We use separate data structures for the sets $\{S_x\}_{x\in [d]}$ and $\{S_y\}_{y\in [k]^d}$.

\begin{lemma}
\label{lem:DSx}
    For any $x\in [d]$, there exists a data structure $\DD_x$ that can perform independence updates for $S_x$ in $O(\log n)$ time with probability at least $1-1/n$.
\end{lemma}
\begin{proof}
    By definition, all hyperrectangles in $S_x$ have (scaled) length at most 1 in dimension $x$. Let $\HH'(x)$ denote the set of hyperrectangles selected by $\greedy$; at the beginning of the observation phase, $\HH'(x) = \emptyset$. We store the starting points of the projections of the hyperrectangles of $\HH'(x)$ in dimension $x$, in a range tree. When a hyperrectangle $H\in S_x$ arrives, we perform a range query with the projection of $H$ in dimension $x$. Recall that there are only $O(\log n)$ hyperrectangles whose projection in dimension $x$ overlaps with that of $H$ (\Cref{lem:gaps}). All of these are output by the range query in $O(\log n)$ time. If $H$ does not intersect any of these hyperrectangles, we insert the starting point of $H$ in the range tree, which again takes $O(\log n)$ time. Hence the overall time required for an independence update is $O(\log n)$.
\end{proof}

\begin{lemma}
\label{lem:DSy}
    For any $y\in [k]^d$, there exists a data structure $\DD_y$ that can perform independence updates for $S_y$ in $O_d(\log n)$ time.
\end{lemma}
\begin{proof}
    For each dimension $j\in [d]$, let the side lengths of the hyperrectangles of $S_y$ lie in the range $[l_j,2l_j)$, for some $l_j\in \mathbb{R}$. We construct a $d$-dimensional grid where the grid lines are situated at a separation of $l_j$ in dimension $j$. Each grid cell stores all the hyperrectangles that intersect that cell. We now state two observations that follow using standard volume arguments.

    \begin{observation}
        Each grid cell can intersect at most $5^d = O(1)$ hyperrectangles.
    \end{observation}

    \begin{observation}
        Each hyperrectangle intersects at most $3^d=O(1)$ grid cells.
    \end{observation}

    Hence, the independence update for $H$ is handled as follows: we enumerate all grid cells that intersect $H$ and check the intersection of $H$ with each of the hyperrectangles stored in those cells. Enumerating the grid cells can be done in $O_d(\log n)$ time with the help of $d$ binary search operations. If $H$ does not intersect any of the hyperrectangles, then it is inserted into each of the enumerated grid cells.
\end{proof}

Hence, the overall algorithm during the observation phase is as follows: when a hyperrectangle $H$ arrives, we first compute the scaled side lengths of $H$ in each dimension in order to determine the size class in which it lies. As in the case of intervals, this can be done in $d\cdot O(\log n) = O_d(\log n)$ time. Once the size class is computed, we outsource the hyperrectangle to the corresponding data structure of that size class.

By maintaining a count of the number of hyperrectangles selected by $\greedy$ in each class, and pointers to the classes having the maximum count among $[d]$ and $[k]^d$, the index $m$ to proceed during the action phase can be determined in $O(1)$ time. Finally, depending on whether $m$ belongs to $[d]$ or $[k]^d$, the $\greedy$ algorithm during the action phase can be implemented using the data structures discussed in \Cref{lem:DSx} or \Cref{lem:DSy}, respectively. Since initialising the data structures and counters for $\{S_x\}_{x\in [d]}$ and $\{S_y\}_{y\in [k]^d}$ only takes $O((\log n)^d)$ time, which is dominated by $O_d(n\log n)$ for large $n$, the overall running time of our algorithm is $O_d(n\log n)$. This completes the proof of \Cref{thm:hyperUnbound}.

%%%%%End of Data Structures%%
%%%%%%End of Appendix%%%%%%%%
\end{document}